\documentclass[aps,prd,notitlepage,nofootinbib,superscriptaddress,groupedaddress]{revtex4-2}

\usepackage{mathrsfs}
\usepackage{bm}
\usepackage{amsmath}
\usepackage{amsthm}
\usepackage{amssymb}
\usepackage{graphicx}
\usepackage{esint}
\usepackage{hyperref}
\usepackage{tabularx}

\theoremstyle{plain}

\theoremstyle{plain}

\ifx\proof\undefined
\newenvironment{proof}[1][\protect\proofname]{\par
	\normalfont\topsep6\p@\@plus6\p@\relax
	\trivlist
	\itemindent\parindent
	\item[\hskip\labelsep\scshape #1]\ignorespaces
}{%
	\endtrivlist\@endpefalse
}
\providecommand{\proofname}{Proof}
\fi

\usepackage[table ]{ xcolor}

\usepackage{amssymb}
\usepackage{graphicx,color}

\usepackage{amsmath,amsthm}

\usepackage{mathrsfs}

\usepackage{bm}

\newtheorem{theorem}{Theorem}[section]

\newtheorem{lemma}[theorem]{Lemma}

\def\beq{\begin{equation}}
\def\eeq{\end{equation}}
\newcommand{\be}{\begin{eqnarray}}
\newcommand{\ee}{\end{eqnarray}}
\def\bi{\begin{itemize}}
\def\ei{\end{itemize}}
\def\ba{\begin{array}}
\def\ea{\end{array}}
\def\bfig{\begin{figure}}
\def\efig{\end{figure}}

\def\C{\mathbb{C}}
\def\R{\mathbb{R}}
\def\Z{\mathbb{Z}}

\newcommand{\bA}{{\bar{A}}}

\newcommand{\Su}{\mathrm{SU}(2)}

\def\be{\begin{eqnarray}}
\def\ee{\end{eqnarray}}

\newcommand{\cb}{\mathcal B}
\newcommand{\cc}{\mathcal C}
\newcommand{\cd}{\mathcal D}

\newcommand{\cf}{\mathcal F}

\newcommand{\ch}{\mathcal H}
\newcommand{\ci}{\mathcal I}
\newcommand{\cj}{\mathcal J}

\newcommand{\cl}{\mathcal L}
\newcommand{\cm}{\mathcal M}
\newcommand{\cn}{\mathcal N}
\newcommand{\co}{\mathcal O}
\newcommand{\calp}{\mathcal P}

\newcommand{\cu}{\mathcal U}

\newcommand{\cz}{\mathcal Z}

  \newcommand{\Fd}{\mathfrak{D}}

\newcommand{\fl}{\mathfrak{l}}  
  
\newcommand{\fn}{\mathfrak{n}}

  \newcommand{\Fs}{\mathfrak{S}}

\renewcommand{\a}{\alpha}
\renewcommand{\b}{\beta}
\newcommand{\g}{\gamma}
\newcommand{\G}{\Gamma}

\newcommand{\eps}{\varepsilon}

\newcommand{\sig}{\sigma}
\newcommand{\Sig}{\Sigma}
\renewcommand{\l}{\lambda}
\renewcommand{\L }{\Lambda}

\newcommand{\rmd}{\mathrm d}

\newcommand{\lt}{\left}
\newcommand{\rt}{\right}

\newcommand{\lag}{\left\langle}
\newcommand{\rag}{\right\rangle}

\newcommand{\bbc}{\mathbb{C}}

\newcommand{\Ar}{\mathrm{Ar}}

\newcommand{\re}{\mathrm{Re}}

\newcommand{\Tr}{\mathrm{Tr}}

\begin{document}

\title{Entanglement entropy in Loop Quantum Gravity and geometrical area law}

\author{Muxin Han}
\email{hanm(At)fau.edu}
\affiliation{Department of Physics, Florida Atlantic University, 777 Glades Road, Boca Raton, FL 33431, USA}
\affiliation{Institut f\"ur Quantengravitation, Universit\"at Erlangen-N\"urnberg, Staudtstr. 7/B2, 91058 Erlangen, Germany}



\begin{abstract}
The non-factorizing nature of the Hilbert space in Loop Quantum Gravity (LQG) due to gauge invariance requires a generalized definition of entanglement entropy. This work employs the framework of von Neumann algebras to investigate the entanglement entropy in LQG. On a graph, the holonomy and flux operators within a region and on the boundary generate a non-factor type I von Neumann algebra, which is used to define the entanglement entropy for LQG states. This algebraic formalism is applied to ``fixed-area states''--superpositions of spin networks associated with a surface with a definite macroscopic area given by the LQG area spectrum. By maximizing the entropy, we derive a geometrical area law where the entanglement entropy is proportional to the area. In addition, we show that bulk entanglement can renormalize the area-law coefficient and produce logarithmic corrections. The results in this paper closely relate to LQG black hole entropy. 
\end{abstract}

\maketitle

\tableofcontents

\section{Introduction}

The intersection of quantum gravity, information, and thermodynamics has revealed some of the deepest puzzles in modern theoretical physics \cite{Bekenstein:1973ur,Wald:1999vt,Jacobson:1995ab,Wald:1993nt}. A central concept that ties these fields together is entropy. The discovery that black holes possess a thermodynamic entropy proportional to their horizon area, known as the Bekenstein-Hawking entropy, suggested that spacetime geometry itself might have a microscopic quantum structure. This idea is further reinforced by the behavior of entanglement entropy in quantum field theory. For a quantum system in a pure state, the entanglement entropy of a subregion measures the degree of quantum correlation between the subregion and its complement. In many quantum field theories, this entropy is found to be divergent, with the leading divergence proportional to the area of the boundary separating the two regions \cite{Sorkin:2014kta,Callan:1994py}. This "area law" for entanglement entropy bears a striking resemblance to the black hole entropy formula, motivating the hypothesis that the entropy of a black hole might be understood, at least in part, as the entanglement entropy of quantum degrees of freedom across its horizon. This connection has become a cornerstone of the holographic principle and the bulk-boundary correspondence, where the Ryu-Takayanagi formula relates the bulk gravitational entropy to the entanglement entropy of the boundary conformal field theory \cite{Ryu:2006bv,Lewkowycz:2013nqa}.

However, to fully understand the origin of gravitational entropy, one needs a complete theory of quantum gravity. Loop Quantum Gravity (LQG) is a prominent candidate for such a theory, offering a non-perturbative and background-independent quantization of general relativity \cite{thiemann2008modern,rovelli2004quantum,rovelli2014covariant,han2007fundamental,ashtekar2004background}. In LQG, the fundamental quantum excitations of geometry are not point-like, but one-dimensional, polymer-like structures. The quantum states of space are described by spin networks \cite{rovelli1988knot,rovelli1995discreteness}, which are graphs with links decorated by representations of the gauge group SU(2) and nodes by invariant tensors (intertwiners). Geometric observables, such as area and volume, have discrete spectra, implying a granular structure of space at the Planck scale. The discreteness of the area spectrum, in particular, provides a natural starting point for a microscopic calculation of black hole entropy by counting the number of horizon configurations that yield a given macroscopic area \cite{Rovelli:1996dv,Ashtekar1998,GP2011,QGandBH,Agullo:2010zz,polymer}. It also leads to a derivation of the Ryu-Takayanagi formula in the framework of LQG \cite{HanHung}.

The existing studies of LQG entanglement entropy mostly focus on the spin-network states with fixed spins see e.g. \cite{Bianchi:2023avf,Feller:2017jqx,Bianchi:2018fmq,Hamma:2015xla,Donnelly:2008vx}. The entanglement entropy of these states does not exhibit the area law. Here, by area law, we mean that the entanglement entropy is proportional to the area of the entangling surface given by the area spectrum of LQG quantum geometry. We sometimes also refer to this as the \emph{geometrical} area law, in order to distinguish from the ``topological'' area law: the entanglement entropy is proportional to the number of links in the network intersecting the surface, in e.g. \cite{Bianchi:2018fmq,Pastawski:2015qua,Qi1}.

In this paper, we analyze the entanglement entropy of LQG states formed by summing spin-networks over different spins. We demonstrate that, for a wide class of such states, the entanglement entropy exhibits the geometrical area law. For states involving spin sums, defining and calculating entanglement entropy presents a fundamental challenge. When a spatial slice is partitioned into a region $A$ and its complement $\bar{A}$, the gauge invariance inherent to the theory leads to a non-factorizing Hilbert space. The degrees of freedom on the boundary $\partial A$—the spins $j_{\fl_0}$ on the links $\fl_0$ piercing the entangling surface—are shared between the two subregions. As a result, the total Hilbert space does not decompose into a simple tensor product $\mathcal{H}_A \otimes \mathcal{H}_{\bar{A}}$, but rather into a direct sum over ``superselection sectors'' labelled by the boundary data $\alpha\equiv \{j_{\fl_0}\}_{\fl_0}$: $\mathcal{H} = \bigoplus_\alpha (\mathcal{H}_{A,\alpha} \otimes \mathcal{H}_{\bar{A},\alpha})$. This structure invalidates the standard definition of entanglement entropy based on the partial trace. This issue is not unique to LQG but appears in any gauge theory.

In this situation, there is a more general, algebraic definition of entropy. The appropriate framework is that of von Neumann algebras \cite{Harlow:2016vwg,Sorce:2023fdx}. This paper employs this algebraic approach to define and compute entanglement entropy in LQG. Our proposal, detailed in Sections \ref{finitedimcase} and \ref{VNA Infinite-dimensional Hilbert space}, is the generalized formula for entanglement entropy applicable to non-factorizing Hilbert spaces:
\be
S(\bm{\rho},A)=-\sum_\alpha p_\alpha \log p_\alpha +\sum_\alpha p_\alpha S(\bm{\rho}_{A,\alpha})+\sum_\a p_\a\log\l_\a.\label{entropyformula0000}
\ee
In our discussion, this entropy formula is from a non-factor type-I von Neumann algebra generated by the holonomy and flux operators in LQG. The same formula has been used for studies on the entanglement entropy of gauge theories in e.g. \cite{Lin:2018bud,Bianchi:2024aim}. Here, $p_\alpha$ is the probability of the state $\bm\rho$ being in the sector $\alpha$, and $S(\bm{\rho}_{A,\alpha})$ is the standard von Neumann entropy of the reduced state within that sector. $\l_\a\in\Z_+$ counts the edge modes at the entangling surface. This formula beautifully decomposes the total entanglement into two physically distinct components: a "classical" Shannon entropy term, arising from the observer's ignorance about the state of the boundary, a quantum term, which is the average of the entanglement entropies of the bulk degrees of freedom within each sector, and an edge-mode term due to the observer's ignorance about the gauge invariance at the boundary. The Shannon term, which depends on the distribution of boundary states, plays a central role in the emergence of the area law.

An important step in justifying this proposal is to relate LQG to the formalism of von Neumann algebra, see \cite{DianaKaminski2011} for some early works along this direction. In Section \ref{VNA Infinite-dimensional Hilbert space}, we demonstrate that on a graph, the weak closure of the algebra generated by holonomy and flux operators within region $A$ (and on the boundary) is a non-factor type I von Neumann algebra identical to $\cm = \bigoplus_\alpha (\cb(\ch_{A,\alpha})\otimes \bm{I}_{\bA,\alpha})$, where $\cb(\ch_{A,\alpha})$ denotes the space of all bounded operators on $\ch_{A,\alpha}$. This result provides a rigorous foundation for applying the entropy formula \eqref{entropyformula0000} to LQG. With this framework in place, this paper presents the following main results:

\begin{enumerate}
    \item We introduce a class of "fixed-area states," which are quantum states representing a closed surface $\Fs$ with a definite macroscopic area $\Ar(\Fs)$ given by the LQG area spectrum. The state can have arbitrarily many links piercing $\Fs$. For such states with no entanglement between the bulk degrees of freedom on either side of the boundary, the entanglement entropy is given purely by the Shannon term. Maximizing this entropy leads to a result given by counting microscopic states on the surface $\Fs$ similar to the LQG black hole entropy computation. As a result, the entanglement entropy is directly proportional to the area of the entangling surface. This yields a geometrical area law, $S \propto \mathrm{Ar}(\Fs)$, where the coefficient of proportionality is calculated explicitly. This derivation provides a direct link between the quantum geometry of LQG and the area-law behavior of entanglement entropy. Our result suggests a relation between the entanglement of LQG states and the classical geometry. This quantum-to-classical relation is independent of the large-$j$ limit, which is often used in the semiclassical analysis of LQG. Indeed, the area-law entropy is dominated by contributions from small $j$'s but large number of punctures on $\Fs$.
    
    \item We investigate the contribution of entanglement between the bulk degrees of freedom. When such entanglement is introduced, the quantum term in the entropy formula becomes non-zero. For certain classes of states where the bulk entanglement can be factored per puncture, its effect is to introduce a spin-dependent degeneracy factor $G(j)$ into the microstate counting. This leads to a "renormalization" of the coefficient of the area law. This mechanism may be viewed as analogous to the renormalization of Newton's constant by quantum field entanglement in the vicinity of a black hole horizon \cite{Larsen:1995ax,Kabat:1995eq,Jacobson:1994iw}.

    \item We show that more complex correlations in the bulk can lead to logarithmic corrections to the area law. By modeling the bulk state in region $A$ as a single large intertwiner, we find the entropy takes the form $S = \beta \mathrm{Ar}(\Fs) +c\log(\mathrm{Ar}(\Fs)) + \dots$, where the coefficient of the area term is renormalized and a specific logarithmic correction appears. The resulting $\b$ and $c$ matches the known result of SU(2) black hole entropy \cite{Agullo:2009eq,Engle2011}, solidifying the connection between entanglement and black hole thermodynamics in LQG. 
    
\end{enumerate}

The paper is organized as follows. Section \ref{The Hilbert space} reviews the LQG Hilbert space on a graph and its non-factorizing structure. Section \ref{finitedimcase} introduces the von Neumann algebra formalism, focusing on the finite-dimensional case and deriving the entropy formula \eqref{entropyformula0000} for $\l_\a=1$. In Section \ref{Fixed-area state and geometrical area law}, we analyze the fixed-area states to derive the area law. Section \ref{A general framework of state-counting} presents a general mathematical scheme for the state-counting, based on inverse Laplace transform techniques. Section \ref{Bulk entropy} explores the effects from bulk entanglement. Section \ref{VNA Infinite-dimensional Hilbert space} generalizes the algebraic framework to the infinite-dimensional case, establishing its validity for the full theory and relating $\l_\a$ to the trace normalization of von Neumann algebra, as well as discussing the physical meaning.








\section{The Hilbert space}\label{The Hilbert space}

The analysis in this paper is based on a closed nontrivial graph $\Gamma_0$ embedded in the compact 3d space $\Sigma$. The graph $\G_0$ contains finite numbers of oriented links $\fl$ and nodes $\fn$. All nodes have their valences greater than 2. The quantum state of LQG along each link $\fl$ belongs to the Hilbert space 
\be
\ch_\fl\simeq L^2(\mathrm{SU}(2),\rmd\mu_H)\simeq \bigoplus_{j\in\mathbb{N}_0/2}\lt(V_j\otimes V_j^*\rt)
\ee
where $j=0,1/2,1,\cdots$ denotes the SU(2) spin. We associate $V_j$ to the source of the link and associate $V_j^*$ to the target. The LQG quantum state on $\G_0$ belongs to the tensor product
\be
\ch_{\G_0}^{\rm (aux)}:=\bigotimes_{\fl\subset\G_0}\ch_\fl\simeq \bigoplus_{\{j_\fl\}}\lt[\bigotimes_{\fn\in\G_0}\widetilde{\ch}_\fn\lt(\{j_\fl\}\rt)\rt],\qquad \widetilde{\ch}_\fn\lt(\{j_\fl\}\rt)=\bigotimes_{\fl,\ \fn=s(\fl)} V_{j_\fl}\otimes \bigotimes_{\fl',\ \fn=t(\fl')} V^*_{j_{\fl'}}.\label{htilde}
\ee
The SU(2) gauge invariance projects each $\widetilde{\ch}_\fn\lt(\{j_\fl\}\rt)$ to the invariant subspace
\be
\ch_\fn\lt(\{j_\fl\}\rt)=\mathrm{Inv}_{\Su}\lt(\bigotimes_{\fl,\ \fn=s(\fl)} V_{j_\fl}\otimes \bigotimes_{\fl',\ \fn=t(\fl')} V^*_{j_{\fl'}}\rt).
\label{eq:Hn}
\ee
which is also called the space of SU(2) intertwiners $i_\fn$ at the node $\fn$. The intertwiner space is 1-dimensional: $\ch_\fn\simeq \C$, for $\fn$ whose adjacent links all have trivial spins. The intertwiner space $\ch_\fn$ is 0-dimensional if the spins on the adjacent links do not satisfy the triangle inequalities. 

After the gauge invariant projection, the Hilbert space $\ch_{\G_0}$ of SU(2) gauge invariant states is defined by
\be
\ch_{\Gamma_0} \simeq \bigoplus_{\{j_\fl\}}\lt[ \bigotimes_{\fn\in\G_0 } \ch_\fn\lt(\{j_\fl\}\rt)\rt].\label{directsum}
\ee
The direct sum contains the trivial spins $j_\fl=0$. A state with $j_\fl=0$ at some links $\fl$ is equivalent to a state with non-trivial spins on a smaller graph $\G\subset\G_0$. $\ch_{\G_0}$ is spanned by the finite linear combinations of spin-networks labelled by $(\G,\{j_\fl\},\{i_\fn\})$ for all $\G\subseteq\G_0$ and all spins $j_\fl$ and intertwiners $i_\fn$ on $\G$.


We divide the spatial slice $\Sig$ into subregions $A$ and $\bA$, and we will consider the entanglement of states in $\ch_{\G_0}$ according to this division. The interface $\Fs$ between $A$ and $\bA$ is assumed to be a closed surface.

In this work, we assume that the surface $\Fs$ intersects with a number of links of the graph $\G_0$, whereas $\Fs$ does not contain any node in $\G_0$. The links that intersects with $\Fs$ are denoted by $\fl_0$. The intersection between a link $\fl_0$ and $\Fs$ is called a puncture. We denote by $\a$ a profile of spins on all intersected links $\fl_0$, i.e. $\a \equiv \{j_{\fl_0}\}_{\fl_0\cap \Fs\neq \emptyset}$. Note that $j_{\fl_0}=0$ is allowed in $\a$.

If we fix a spin profile $\{j_\fl\}$ on all links of the entire $\G_0$, $\{j_\fl\}$ determines $\a$ on $\{\fl_0\}$, the corresponding subspace $\bigotimes_{\fn\in\G_0 } \ch_\fn\lt(\{j_\fl\}\rt)\subset \ch_\G$ can be factorized into $A$ and $\bA$
\be
\ch_{A}\lt(\a, \{j_\fl\}_{\fl\subset A}\rt) \otimes \ch_{\bA  }\lt(\a, \{j_\fl\}_{\fl\subset \bA}\rt)
\ee
where $\ch_{A}\lt(\a, \{j_\fl\}_{\fl\subset A}\rt) $ and $ \ch_{\bA  }\lt(\a, \{j_\fl\}_{\fl\subset \bA}\rt)$ associate to the subregions $A$ and $\bA$ respectively:
\be
\ch_{A}\lt(\a, \{j_\fl\}_{\fl\subset A}\rt) = \bigotimes_{\fn\in\G_0 \cap A } \ch_\fn\lt(\{j_\fl\}\rt),\qquad
\ch_{\bA  }\lt(\a, \{j_\fl\}_{\fl\subset \bA}\rt)=\bigotimes_{\fn\in\G_0 \cap \bA } \ch_\fn\lt(\{j_\fl\}\rt).
\ee
where $\{j_\fl\}$ on the right hand sides coincide with $\a$ for $\fl=\fl_0$.

However, the factorization of Hilbert space into degrees of freedoms of $A$ and $\bA$ cannot be achieved at the level of $\ch_{\G_0}$ with the spins summed over. The factorization only holds when the spins on $\{\fl_0\}$ are fixed. Indeed, the direct sum over spins in \eqref{directsum} can be split into the sums over the spins inside $A$, $\bA$ and over $\a$ on $\{\fl_0\}$:
\be
\ch_{\G_0}=\bigoplus_{\a }\lt(\ch_{A,\a} \otimes \ch_{\bA,\a }\rt),\label{chPj}
\ee
where
\be
\ch_{A,\a }=\bigoplus_{\{j_\fl\}_{\fl\subset A}}\ch_{A}\lt(\a, \{j_\fl\}_{\fl\subset A}\rt),\qquad \ch_{\bA,\a}=\bigoplus_{\{j_\fl\}_{\fl\subset \bA}}\ch_{\bA}\lt(\a, \{j_\fl\}_{\fl\subset \bA}\rt).
\ee
The sum of $\{j_\fl\}$ in $\ch_{A,\a }$ and $\ch_{\bA,\a } $ is only over the internal $j_\fl$ inside $A$ and $\bA$ and excludes $j_{\fl_0}$. The intertwiner space $\ch_{\fn}(\{j_\fl\})$ in $\ch_{A,\a}$ or $\ch_{\bA,\a}$ for $\fn$ adjacent to $\Fs$ has the fixed spins $j_{\fl_0}$ for $\fl_0$ intersecting $\Fs$.


Once the sum over spins are taken into account, the Hilbert space $\ch_{\G_0}$ is not a tensor product of $\ch_A$ and $\ch_{\bA}$, but a direct sum of tensor products. The usual definition of the reduced density matrix and entanglement entropy for a space region $A$ requires the factorization Hilbert space $\ch\simeq \ch_A\otimes \ch_\bA$, so that the partial trace on $\ch_{\bA}$ can be defined. However, the LQG Hilbert space clearly does not satisfy this requirement. Therefore, the usual notion of entanglement entropy needs to be generalized in order to be applied to LQG. The generalization uses the language of non-factor von Neumann algebra and will be discussed in Section \ref{finitedimcase} and the following sections (see also e.g. \cite{Colafranceschi:2023moh,Harlow:2016vwg}).

The non-factorization of the Hilbert space is a consequence of the gauge invariance on the surface $\Fs$: Consider a link $\fl_0$ cut by $\Fs$ into a pair of half-edges. The gauge invariance at the added node in between the half-edges constrains the spins on the half-edges to be the same. Then $j_{\fl_0}$ becomes a degree of freedom shared by both $A$ and $\bA$. If a Hilbert space can be factorized, all degrees of freedom has to either belong to $A$ or belong to $\bA$, but this contradicts with the property of $j_{\fl_0}$. We will come back to the discussion of this aspect in Section \ref{VNA Infinite-dimensional Hilbert space}.

In the following, we often denote by $\calp$ the set of punctures on $\Fs$ carrying nontrivial spins and denote by $\cj=\{j_{\fl_0}\neq 0\}_{\fl_0\cap\Fs\neq \emptyset}$ the profile of nontrivial spins on $\calp$. Given the graph $\G_0$, we have the equivalence $\a=(\calp,\cj)$.

The full LQG kinematical Hilbert space extends $\ch_{\G_0}$ by including all graphs $\G$ in addition to the ones included by $\G_0$. However, in this paper, we focus on a fixed graph $\G_0$ and the Hilbert space $\ch_{\G_0}$. In the following, we neglect the label and denote by $\ch=\ch_{\G_0}$.



\section{Von Neumann algebra and entanglement entropy: Finite-dimensional Hilbert space}\label{finitedimcase}

In this section, we discuss the entanglement entropy from the viewpoint of von Neumann algebra on the Hilbert space of the type 
\be
\hat{\ch}\simeq \bigoplus_{\alpha}\lt(\hat{\ch}_{A,\alpha } \otimes \hat{\ch}_{\bA,\alpha }\rt).\label{decomposition1}
\ee
For simplifying the discussion, we first assume $\hat{\ch}$, $\hat{\ch}_{A,\alpha }$, $\hat{\ch}_{\bA,\alpha }$ are all finite-dimensional in this section and postpone the discussion for the infinite-dimensional Hilbert space to Section \ref{VNA Infinite-dimensional Hilbert space}. In the case of \eqref{chPj}, the finite-dimensional Hilbert space corresponds to imposing truncations to the spin-sums for all $\fl$. The maximum of spins may have the physical meaning as implementation of cosmological constant: $j_{\rm max}\sim \L^{-1}/(8\pi\g \ell_P^2)$. 



Given that $\hat{\ch}$ is a finite-dimensional Hilbert space and $\cl(\hat{\ch})$ is the space of linear operators on $\hat{\ch}$. A von Neumann algebra $\cm\subset\cl(\hat{\ch})$ is a subalgebra of operators which is closed under hermitian conjugation and contains the identity operator. The algebra $\cm$ may have the commutant $\cm'\subset\cl(\hat{\ch})$ whose operators commute with all operators in $\cm$. $\cz_\cm=\cm\cap\cm'$ is the center of $\cm$. A von Neumann algebra on a Hilbert space is called a factor if its center contains only scalar multiples of identity operator. On finite-dimensional Hilbert space, if a von Neumann algebra is a factor, then it is of type I (see e.g. \cite{Sorce:2023fdx} for the classification of von Neumann algebra factors).

Given any von Neumann algebra $\cm$ on $\hat{\ch}$, which is not necessary a factor, there exits a block decomposition $\hat{\ch}=\oplus_\alpha (\hat{\ch}_{A,\alpha}\otimes \hat{\ch}_{\bA,\alpha})$, such that $\cm=\oplus_\alpha (\cl(\hat{\ch}_{A,\alpha})\otimes \bm{I}_{\bA,\alpha})$ and $\cm'=\oplus_\alpha ( \bm{I}_{A,\alpha}\otimes\cl(\hat{\ch}_{\bA,\alpha}))$ \cite{Harlow:2016vwg}. When $\cm$ is a factor, the direct sum of $\alpha$ becomes trivial.

In the case of LQG, we make this block decomposition coincide with the decomposition of $\ch$ in \eqref{chPj} by identifying the von Neumann algebra $\cm $ to be $\oplus_\alpha (\cl(\hat{\ch}_{A,\alpha})\otimes I_{\bA,\alpha})$ with $\alpha=(\calp,\cj)$ and $\hat{\ch}_{A,\alpha}$ being the truncation of $\ch_{A,(\calp,\cj)}$ in \eqref{chPj} (by the cut-off of the spins). It is clear that $\cm$ defined in this way contains the SU(2) gauge invariant operators made by the holonomies and fluxes in region $A$ which are
\begin{itemize}

\item The holonomy operators $\bm{h}^{(s)}_\fl$ for $\fl\subset A$ and importantly $\fl\neq \fl_0$, i.e. The holonomies do not intersect with $\Fs$. $s\in\mathbb{N}/2$ denotes the representation of the holonomy.

\item The flux operator that gives the angular momentum operators $\vec{\bm{R}}_\fl,\vec{\bm{L}}_\fl$ acting on $V_{j_\fl}$ and $V_{j_\fl}^*$ for $\fl\subset A$ and $\fl\neq \fl_0$.

\item The angular momentum operator $\vec{\bm{J}}_{\fl_0}=\vec{\bm{R}}_{\fl_0}$ or $\vec{\bm{L}}_{\fl_0}$ acting on $V_{j_{\fl_0}}$ or  $V_{j_{\fl_0}}^*$ for $\fl_0\cap \Fs\neq \emptyset$, where $V_{j_{\fl_0}}$ or  $V_{j_{\fl_0}}^*$ belongs to an intertwiner space in $\ch_{A,\alpha}$. In other words, $\vec{\bm{J}}_{\fl_0}$ acts on the open legs of the intertwiners inside $A$. 

\end{itemize}
The gauge invariance is implemented at the nodes $\fn$ inside $A$. Given any gauge invariant operators $\bm{\co}$ made by holonomies and fluxes, the corresponding element in $\cm$ is given by $\hat{\bm{p}}\bm{\co}\hat{\bm{p}}$, where $\hat{\bm{p}}$ is the projection from $\ch$ to the truncated Hilbert space $\hat{\ch}$. In Section \ref{VNA Infinite-dimensional Hilbert space}, we will show on the infinite-dimensional Hilbert space $\ch$ that the algebra generated by the holonomies and fluxes in $A$ is isomorphic to the von Neumann algebra $\oplus_\alpha (\cl({\ch}_{A,\alpha})\otimes I_{\bA,\alpha})$. In addition, the area operator associated to every puncture, $\bm{\Ar}_{\fl_0}=8\pi\g\ell_P^2\sqrt{\vec{\bm{J}}_{\fl_0}\cdot \vec{\bm{J}}_{\fl_0}}$, belongs to the center $\cz_\cm$.

Given any density matrix $\bm{\bm{\rho}}\in\cl(\hat{\ch})$ normalized by $\Tr(\bm{\bm{\rho}})=1$ and a von Neumann algebra $\cm$ on $\hat{\ch}$, there exists a unique Hermitian, non-negative, normalized $\bm{\bm{\rho}}_\cm\in\cm $ such that $\Tr(\bm{\bm{\rho}} \bm{\co})=\Tr(\bm{\bm{\rho}}_\cm \bm{\co})$ for all $\bm{\co}\in \cm$ (see e.g. \cite{Harlow:2016vwg} for a proof of this statement). Explicitly, $\bm{\bm{\rho}}$ can be written as blocks $\bm{\bm{\rho}}_{\a\a'}$ according to \eqref{decomposition1}. Let us only focus on the diagonal block $\bm{\bm{\rho}}_{\a\a}$, since only the diagonal blocks contributes to the expectation values $\Tr(\bm{\bm{\rho}} \bm{\co})$ of the operators $\bm\co\in \cm$. All operators $\bm{\co}\in \cm$ are block diagonal since $\cm=\bigoplus_\alpha \lt(\mathcal{L}(\hat{\ch}_{A,\alpha})\otimes \bm{I}_{\bA,\alpha}\rt)$. Each diagonal block $\bm{\bm{\rho}}_{\a\a}$ is on the factorized subspace $\hat{\ch}_{A,\a}\otimes\hat{\ch}_{\bA,\a}$. We take the partial trace on $\hat{\ch}_{\bA,\alpha}$ and define
\be
 \bm{\rho}_{A,\a}=p_\alpha^{-1}\Tr_{\bA,\alpha}(\bm{\rho}_{\a\a})
\ee
where $0<p_\alpha<1$ is the normalization constant such that $\Tr_{A,\a}( \bm{\rho}_{A,\a})=1$. Then $\bm{\rho}_\cm$ is given by
\be
\bm{\rho}_\cm=\bigoplus_\alpha \lt(p_\alpha \bm{\rho}_{A,\a}\otimes \frac{\bm{I}_{\bA,\alpha}}{\dim(\hat{\ch}_{\bA,\alpha})} \rt).\label{rhoM}
\ee
where $I_{\bA,\a}$ is the identity operator on $\ch_{\bA,\alpha}$. This follows from the uniqueness of $\bm{\rho}_{\cm}$ and the fact that the above expression of $\bm{\rho}_\cm$ satisfies all the required properties (i.e. Hermitian, non-negative, normalized, and  $\Tr(\bm{\rho} \bm{\co})=\Tr(\bm{\rho}_\cm \bm{\co})$).

When the direct sum of $\alpha$ is trivial, i.e. $\hat{\ch}$ factorizes as $\hat{\ch}_A\otimes\hat{\ch}_\bA$, the state $\bm{\rho}_{\cm}$ relates to the reduced density matrix $\bm{\rho}_A$ by $\bm{\rho}_\cm=\bm{\rho}_A\otimes I_{\bA}/\dim(\ch_{\bA})$, and the von Neumann entropy is given by $-\Tr_A(\bm{\rho}_A\log\bm{\rho}_A)$. However, there is no notion of $\bm{\rho}_A$ and $\Tr_A$ from the non-factorized $\hat{\ch}$ in \eqref{decomposition1}. The idea is to define the entropy with $\bm{\rho}_\cm$ instead of $\bm{\rho}_A$. We also need to define a renormalized trace on $\cm$ \cite{Harlow:2016vwg,ohya2004quantum}. Let us firstly consider the case of trivial sum over $\alpha$ ($\cm$ is a factor) then generalize from there: For $\hat{\ch}=\hat{\ch}_A\otimes\hat{\ch}_\bA$, we define the renormalized trace $\hat{\Tr} = \frac{1}{\dim(\hat{\ch}_\bA)}\Tr$. For any operator $\bm{\co}\in\cm$, we have $\bm{\co}=\bm{\co}_A\otimes \bm{I}_\bA$, where $\bm{\co}\in\cl(\hat{\ch}_A)$ so $\hat{\Tr}(\bm{\co})=\Tr_A(\bm{\co}_A)$. Since $\hat{\ch}$ is finite-dimensional and $\cm$ is type-I, the minimal projection in $\cm$ is given by $\bm{P}=\bm{P}_{A}\otimes \bm{I}_{\bA}$, where $\bm{P}_{A}$ is a projection onto 1-dimensional subspace of $\hat{\ch}_{A}$. The above renormalized trace satisfies $\hat{\Tr}(\bm{P})=1$. We also renormalize $\bm{\rho}_\cm$ by $\hat{\bm{\rho}}_\cm=\dim(\hat{\ch}_\bA){\bm{\rho}}_\cm=\bm{\rho}_A\otimes \bm{I}_\bA$ so that $\hat{\Tr}(\hat{\bm{\rho}}_\cm\bm{\co})=\Tr(\bm{\rho}_\cm\bm{\co})=\Tr_A(\bm{\rho}_A\bm{\co}_A)$. The entanglement entropy is defined by $S(\bm{\rho},A):=-\hat{\Tr}(\hat{\bm{\rho}}_\cm\log \hat{\bm{\rho}}_\cm)$, which reduces to the standard definition of von Neumann entropy $- {\Tr}({\bm{\rho}}_A\log {\bm{\rho}}_A)$.

When the sum of $\alpha$ is nontrivial ($\cm$ is not a factor), The renormalized trace $\hat{\Tr}$ on $\hat{\ch}=\oplus_\a(\hat{\ch}_{A,\a}\otimes\hat{\ch}_{\bA,\a})$ is defined by (see e.g. \cite{Harlow:2016vwg})
\be
\hat{\Tr}:=\bigoplus_\alpha\frac{1}{\dim(\hat{\ch}_{\bA,\alpha})}\Tr\ ,\label{renTracefinite}
\ee
Any operator $\bm{\co}\in\cm$ is given by $\bm{\co}=\oplus_\alpha (\bm{\co}_{A,\alpha}\otimes\bm{I}_{\bA,\alpha})$ where $\bm{\co}_{A,\alpha}\in\cl(\hat{\ch}_{A,\alpha})$, then the action of $\hat{\Tr}$ is
\be
\hat{\Tr}(\bm{\co})=\sum_\alpha \Tr_{A,\alpha }(\bm{\co}_{A,\alpha}).
\ee
We have again $\hat{\Tr}(\bm{P}_\a)=1$ for any minimal projection $\bm{P}_\a\in \cm$, where $\bm{P}_\a=\bm{P}_{A,\a}\otimes\bm{I}_{\bA,\a}$ for some $\alpha$ and the projection $\bm{P}_{A,\a}$ onto 1-dimensional subspace of $\hat{\ch}_{A,\a}$. The trace $\hat{\Tr}$ satisfies the property $\hat{\Tr}(\bm{x}\bm{y})=\hat{\Tr}(\bm{y}\bm{x})$ for all $\bm{x},\bm{y}\in\cm$. We also renormalize $\bm{\rho}_\cm$ in \eqref{rhoM} by 
\be
\hat{\bm{\rho}}_\cm:=\bigoplus_\alpha \lt(p_\alpha \bm{\rho}_{A,\a}\otimes  {\bm{I}_{\bA,\alpha}} \rt),
\ee
so that $\hat{\Tr}(\hat{\bm{\rho}}_\cm)=\Tr(\bm{\rho}_\cm)=1$ and for any $\bm \co\in \cm$, we have $\hat{\Tr}(\hat{\bm{\rho}}_\cm\bm \co)=\Tr(\bm{\rho}_\cm\bm \co)$. The entanglement entropy is defined by
\be
S(\bm{\rho},A):=-\hat{\Tr}\lt(\hat{\bm{\rho}}_\cm\log\hat{\bm{\rho}}_\cm\rt)=-\sum_\alpha p_\alpha \log p_\alpha +\sum_\alpha p_\alpha S(\bm{\rho}_{A,\alpha}).\label{entropyformula}
\ee
The first term of $S(\bm{\rho},A)$ is the Shannon entropy of the probably distribution $p_\alpha$, while the second term averages the von Neumann entropies over all sectors $\alpha$. This formula of entanglement entropy is defined for states $\bm{\rho}$ on a Hilbert space that is not factorized as tensor product, while it clearly reduces to the standard von Neumann entropy in the case that the Hilbert space factorizes. 


The entropy formula \eqref{entropyformula} is proposed in \cite{Harlow:2016vwg} for the Hilbert space of quantum error correction code. Some basic properties of $S(\bm{\rho},A)$ can also be found in \cite{Harlow:2016vwg}. The same entropy formula is obtained independently in \cite{Hamma:2015xla} from LQG perspective. See also e.g. \cite{Bianchi:2019stn,Bianchi:2021aui} for some applications of this formula. As we are going to see in the following, the Shannon-entropy term in $S(\bm{\rho},A)$ closely relates to the geometrical area law for LQG states.

Note that the definition of the renormalized trace \eqref{renTracefinite} has the freedom of rescaling the trace at each $\a$ by a positive constant $\l_\a$. The trace of the minimal projection becomes $\hat{\Tr}(\bm{P}_\a)=\l_\a$. In the following, we proceed our discussion with the definition \eqref{renTracefinite} and discuss this ambiguity in Section \ref{VNA Infinite-dimensional Hilbert space}.

\section{Fixed-area state and geometrical area law}\label{Fixed-area state and geometrical area law}

\subsection{Fixed-area states}\label{Fixed-area state}

Recall the Hilbert space 
\be
\ch=\bigoplus_{\a}\lt(\ch_{A,\a } \otimes \ch_{\bA,\a  }\rt),\qquad \a=(\calp,\cj)\label{chPj11}
\ee
for LQG. $\calp$ denotes the set of punctures with nontrivial spins on the surface $\Fs$, and $\cj=\{j_{\fl_0}\neq0\}_{\fl_0\cap \Fs\in\calp}$ denotes the profile of spins on the punctures. In this section, we discuss a class of pure states $\bm{\rho}=|\Psi_{\Ar(\Fs)}\rangle\langle\Psi_{\Ar(\Fs)}|$, where $\Psi_{\Ar(\Fs)}\in{\ch}$ endows the area $\Ar(\Fs)$ to $\Fs$ by approximating the eigenstate of the LQG area operator associated to $\Fs$. We are going to show that each of these states gives the entanglement entropy $S(\bm{\rho},A)$ proportional to $\Ar(\Fs)$.

Let us first consider the family of states $\Psi_{\Ar(\Fs)}\in \ch$ satisfying the following conditions:

\begin{enumerate}

\item $\Psi_{\Ar(\Fs)}$ sums the states with different sets $\calp$ of punctures on $\Fs$. Any two different $\calp$ in the sum have different number of nontrivial punctures. We denote by $N$ the number of punctures $|\calp|$. $N\gg1$ is always assumed.

\item $\Psi_{\Ar(\Fs)}$ is a linear combination of the pure tensor-product states in $\ch_{A,\a } \otimes \ch_{\bA,\a  }$, i.e. there is no entanglement of bulk degrees of freedom (e.g. spins and intertwiners inside $A$ and $\bA$) between $A$ and $\bA$ at each $\a$-sector.

\item $\Psi_{\Ar(\Fs)}$ approximates an eigenstate of the area operator associated to $\Fs$ by summing over $(\calp,\cj)$ satisfying
\be
\Ar(\Fs)-\Delta\leq 8\pi\g\ell_P^2\sum_{\fl_0=1}^N\sqrt{j_{\fl_0}(j_{\fl_0} +1)}\leq \Ar(\Fs)+\Delta,\label{areashell}
 \ee
where the width $\Delta$ is of the order $\ell_P^2$.

\end{enumerate}

The states satisfying the above conditions can be written generally as
\be
\Psi_{\Ar(\Fs)}&=&\sum_{\alpha=(N,\cj)} p_\alpha^{1/2} \lt(\psi_{A,\alpha}\otimes\tilde{\psi}_{\bA,\alpha}\rt)\chi_{[-\Delta,\Delta]}\lt(C_1(\a)\rt),\label{Psisum1}\\
C_1&=&8\pi \ell_P^2\g\sum_{\fl_0=1}^N\sqrt{j_{\fl_0}(j_{\fl_0} +1)}-{\Ar(\Fs)}.
\ee
where $p_\alpha>0$. $\chi_{[-\Delta,\Delta]}(x)=1$ for $x\in [-\Delta,\Delta]$ and vanishes elsewhere. Both $\psi_{A,\alpha} $ and $ \tilde{\psi}_{\bA,\alpha}$ are normalized in $\ch_{A,\alpha}$ and $\ch_{\bA,\alpha}$, so $\Vert\Psi_{\Ar(\Fs)}\Vert=1$ implies $\sum_\alpha p_\alpha=1$. The sum over $\calp$ becomes the sum over the number $N$ of punctures due to the first condition. There are only finitely many nonzero terms in the sum due to the bound of area, so the state may be understood as living in a finite-dimensional truncated Hilbert space $\hat{\ch}$.

\subsection{Maximizing the entanglement entropy and microstate-counting}

Given $\bm{\rho} =|\Psi_{\Ar(\Fs)}\rangle\langle\Psi_{\Ar(\Fs)} |$, the diagonal block $\bm{\rho}_{\a\a}$ gives
\be
\Tr_{\bA,\alpha}(\bm{\rho}_{\a\a})=p_\alpha \bm{\rho}_{A,\alpha},\qquad \bm{\rho}_{A,\alpha}=|\psi_{A,\alpha}\rangle\langle\psi_{A,\alpha}|,\qquad S(\bm{\rho}_{A,\alpha})=0,
\label{eq:rho_pure}
\ee
for all $\a = (N,\cj)$. Therefore, the entanglement entropy $S(\bm{\rho},A)$ only receive the contribution from the Shannon term:
\be
S(\bm{\rho},A)=-\sum_{\alpha=(N,\cj)}  p_\alpha \log p_\alpha,\label{entropyplogp}
\ee
where the sum is constrained by \eqref{areashell}. 

In the family of states \eqref{Psisum1} specified by the conditions 1-3, we look for the state that maximizes the entanglement entropy. The entropy $S(\bm{\rho},A)$ in \eqref{entropyplogp} only depends on the coefficients $p_\a$ and the total area $\Ar(\Fs)$ of the state. If we fix the total area $\Ar(\Fs)$, the maximization can be computed by the variation $p_\a\to p_\a+\delta p_\alpha$ under the normalization constraint $\sum_\alpha p_\alpha=1$:
\be
\delta \lt[S(\bm{\rho},A)+\sig\sum_\alpha p_\alpha\rt]=-\sum_{\alpha}\delta p_\alpha \log p_\alpha +(\sig -1)\sum_\alpha \delta p_\alpha=0.\label{maximization1}
\ee
where $\sig$ is a Lagrangian multiplier. The saddle point corresponds to the constant $p_\a$. The normalization constraint implies
\be
p_\alpha =\cd^{-1}.
\ee
$\cd$ is the total number of $\a=(N,\cj)$'s satisfying the area constraint
\be
C_1=8\pi \ell_P^2\g\sum_{\fl_0=1}^N\sqrt{j_{\fl_0}(j_{\fl_0} +1)}-{\Ar(\Fs)}\in [-\Delta,\Delta].
\ee
The entanglement entropy at the saddle point reduces to
\be
S(\bm{\rho}_c,A)=\log\cd\ .\label{logcd}
\ee
The state $\bm{\rho}_c$ with constant $p_\a$ maximizes the entanglement entropy \eqref{entropyplogp}, since the second derivative of $S(\bm{\rho},A)$ gives $-p_\a^{-1}=-\cd <0$.

The formula \eqref{logcd} reduces the entanglement entropy to counting the microstates $\alpha$ that satisfy the area constraint. The computation of $\cd$ is the same as the LQG black hole microstate counting up to a possibly difference of the degeneracy $G(j)$ at spin $j$ (see below). See e.g. \cite{GP2011,Ghosh:2004wq,QGandBH} for the black hole state-counting. It is interesting to see that once the summing graphs and spins are taken into account, the entanglement entropy of LQG states relates to the black hole entropy. 

We define $n_j$ to be the number of punctures carrying the spin $j\neq 0$. Any $\a=(N,\cj)$ gives a set $\{n_j\}_{j\neq 0}$ satisfying $\sum_{j\neq 0} n_j=N$. Conversely, given a set $\{n_j\}_{j\neq 0}$, the total number of nontrivial punctures $N$ is fixed $\sum_{j\neq 0} n_j=N$, so the number of microstates $\a$ is the number of profiles $\cj$ with the same $\{n_j\}_{j\neq 0}$. The corresponding number of microstates $d[\{n_j\}]$ is given by
\be
d[\{n_j\}]=\lt(\sum_{j\neq 0} n_j \rt)!\, \prod_{j\neq0}\frac{G(j)^{n_j}}{n_j!}
\ee
The degeneracy is $G(j)=1$. Here we keep $G(j)$ in the formula for some generalizations to be discussed later. The total number of microstates $\cd$ is given by
\be
\cd = \tilde{\sum_{\{n_j\}}} d[\{n_j\}],\label{Dstatecount}
\ee
where the sum $\tilde{\sum}_{\{n_j\}}$ is subjected to the area constraint.



In the following, we first give a simple computation of $\cd$ using the heuristic method in statistical mechanics with Stirling's approximation and show $\log \cd$ to be proportional to $\Ar(\Fs)$. Then we provide a more rigorous derivation of the same result in Section \ref{A general framework of state-counting}.

We assume $n_j\gg1$. The number of micro-states $\tilde{\sum}_{\{n_j\}} d[\{n_j\}]$ in the ensemble is dominated by the contribution from the configuration $\{\bar{n}_j\}$ which maximizes $d[\{{n}_j\}]$. $\{\bar{n}_j\}$ is the solution of the variational equation $\delta\log d[\{n_j\}]-\b \delta C_1 /(8\pi\ell_P^2\g) =0$ where $\b$ is the Lagrangian multipliers. Here we neglect the width $\Delta$ of the area constraint and impose $C_1=8\pi\ell_P^2\g\sum_{j\neq 0}n_j\sqrt{j(j+1)}-\Ar(\Fs)\approx 0$ to the variation. We obtain the following most-probable distribution under the Stirling's approximation
\be
\frac{\bar{n}_j}{\sum_{j\neq 0}\bar{n}_j}=e^{-\b\sqrt{j(j+1)}}G(j), \label{boltzmann}
\ee
The normalization determines the value of $\beta$:
\be
\sum_{j\neq 0} G(j) e^{-\b\sqrt{j(j+1)}}=1\qquad\Rightarrow\qquad \beta=\b_0\simeq 1.01527\ \text{for}\ G(j)=1.\label{beta0}
\ee
Insert the result into the area constraint:
\be
\frac{\mathrm{Ar}(\mathfrak{S})}{8\pi\ell_{P}^{2}\gamma}=\sum_{j\neq0}\bar{n}_j\sqrt{j(j+1)}\simeq 1.63066 \sum_{j\neq0}\bar{n}_{j},\label{areaconstr1}
\ee
the mean-value of the area quantum number at each puncture is of $O(1)$. It means that the dominant contribution comes from small $j$'s but large number of punctures on $\Fs$.

Note that $e^{-\b\sqrt{j(j+1)}}$ is generally a irrational number, but the left-hand side of \eqref{boltzmann} is rational, so the solution \eqref{boltzmann} is understood as an approximation. Given the area $\Ar(\Fs)$, the constraint \eqref{areaconstr1} is also understood as an approximation. It is the reason why we include the fluctuation $\Delta$ at the beginning. 

Finally, inserting $\{\bar{n}_j\}$ into $d[\{n_j\}]$ gives that the entanglement entropy $S(\bm{\rho}_c,A)$ is proportional to the area $\Ar(\Fs)$:
\be
S(\bm{\rho}_c,A) &=& \log \cd\simeq \log d[\{\bar{n}_j\}]=\b_0\sum_{j\neq 0}\bar{n}_{j}\sqrt{j(j+1)}\nonumber\\
&\simeq & \frac{\beta_0}{8\pi\ell_{P}^{2}\gamma}\mathrm{Ar}(\mathfrak{S})\ .\label{resultS}
\ee

{
In the above counting of the micro-states, we have ignored the possible constraint coming from the intertwiners. Even in the case that there are many nodes in $A$, the intertwiners impose the constraint on the spin profiles $\a$ by requiring the sum of all spins is an integer \footnote{First, the sum of spins of each individual intertwiner is an integer, so the sum of these spin-sums over all intertwiners gives an integer. Clearly the spins in the internal links are over-counted and needs to be subtracted. The number that should be subtracted is again an integer, because each internal spin are over-counted twice. }. It is also possible that the triangle inequalities give some additional constraints. However, these constraints do not modify the area law of the entropy. The explicit evidences are provided in Section \ref{Logarithmic correction} and in \cite{spinfoamstack}. We will also come back to this point in Section \ref{Hidden sector Hilbert space}.
}

\section{A general scheme of state-sum}\label{A general framework of state-counting}

In this section, we provide a general scheme of computing $\cd$ and various generalizations that will appear in the following sections and in \cite{spinfoamstack,spinfoamstack1}. The discussion in this section mostly follows from the computation of LQG black hole entropy in \cite{Agullo:2010zz,Agullo:2009eq,BarberoG:2008dee,Agullo:2008yv,Engle2011}.

We introduce dimensionless area $a=\Ar(\Fs)/(4\pi\gamma\ell_{P}^{2})$ and consider a counting of the number of states whose total area is less or equal to $a$:
\be
\cn_{\leq }(a)=\sum_{N=1}^\infty \sum_{k_1,\cdots,k_N\in\Z_+} P(k)\Theta\lt(a-\sum_{i=1}^N\sqrt{k_i(k_i+2)}\rt),\label{cn(a)}
\ee
where $k=2j$ and $P\left(k\right)=P\left(k_{1},\cdots,k_{N}\right)$ is the degeneracy at each spin configuration. The number of states in $[a-\delta,a+\delta]$, where $\delta=\Delta/(4\pi\gamma\ell_{P}^{2})\sim O(1)$, is given by
\be
\cn_{\leq}\left(a+\delta\right)-\cn_{\leq}\left(a-\delta\right),
\ee
which equals $\cd$ in \eqref{Dstatecount} in the case $P(s)=1$.

Instead of focusing on $P\left(k\right)=1$, let us consider the following ansatz:
\be
P\left(k\right)=\int_\cc \rmd\mu\left(z\right)\prod_{i=1}^Nf\left(k_i,z\right)=\int_\cc \rmd\mu\left(z\right)\prod_{k\in\Z_+} f\left(k,z\right)^{n_k},\label{generalansatz}
\ee
for a certain given function $f(k,z)$. Here $\cc$ is an auxiliary measure space with the positive measure $\rmd\mu(z)$. The key point here is that the integrand of $P(k)$ is factorized into contributions at individual punctures, while each factor $f(k_i,z)$ depending on $i$ only through $k_i$. We assume $f\left(k,z\right)\geq 0$ for $P(k)$ being interpreted as counting the number of states. 


The state-sum in $\cn_{\leq}(a)$ is a finite sum due to the truncation by $\Theta(\cdots)$. By interchanging the sum and integral, we obtain
\be
\cn_{\leq}(a)=\int_\cc \rmd\mu\left(z\right)\sum_{N=1}^\infty \sum_{k_1,\cdots,k_N\in\Z_+} \beta(k)\Theta\lt(a-\sum_{i=1}^N\sqrt{k_i(k_i+2)}\rt),\qquad \beta(k)=\prod_{i=1}^Nf\left(k_i,z\right)
\ee
The sum of the integrand can be computed by the following method of inverse Laplace transform:

\begin{lemma}\label{lemmaV1}
Given two sequence $\{\a_n\}_{n=1}^\infty $ and $\{\b_n\}_{n=1} ^\infty $ where $\b_n\in\C$ and $\a_n>0$, if $\sum_{n=1}^\infty\left|\beta_{n}\right|e^{-\alpha_{n}\re(s)}<\infty$ for some $\re(s)>0$, we have
\be
\sum_{n,\alpha_{n}< a}\beta_{n}=\sum_{n=1}^{\infty}\beta_{n}\Theta\left(a-\alpha_{n}\right) =\frac{1}{2\pi i}\int_{T-i\infty}^{T+i\infty}\frac{\rmd s}{s}\left[\sum_{n=1}^{\infty}\beta_{n}e^{-\alpha_{n}s}\right]e^{as}.\label{invLaplace}
\ee
for the cut-off $a$ that does not coincide with any $\a_n$. The parameter $T>0$ is greater than the real part of all singularities given by the integrand.
\end{lemma}

\begin{proof}
Let us consider the following Laplace transform for $\re(s)>0$
\be
s\int_{0}^{\infty}\rmd a\, e^{-as}\sum_{n=1}^{\infty}\beta_{n}\Theta\left(a-\alpha_{n}\right)=s\sum_{n=1}^{\infty}\beta_{n}\int_{\a_n}^{\infty}\rmd a\, e^{-as}=\sum_{n=1}^{\infty}\beta_{n}e^{-\alpha_{n}s}.
\ee
Interchanging the integral and sum is justified by the Fubini-Tonelli theorem and the following absolutely convergent integral (by the hypothesis):
\be
\sum_{n=1}^{\infty}\int_{0}^{\infty}\rmd a\left|e^{-a s}\beta_{n}\Theta\left(a-\alpha_{n}\right)\right|=\sum_{n=1}^{\infty}\frac{\left|\beta_{n}\right|}{\mathrm{Re}\left(s\right)}e^{-\alpha_{n}\re(s)}.\nonumber
\ee
Moreover, given $\re(s)>0$ such that $\sum_{n=1}^\infty\left|\beta_{n}\right|e^{-\alpha_{n}\re(s)}$ converges, we have the uniform bound
\be
\sum_{n,\a_n<a}|\beta_n|e^{-a\re(s)}\leq \sum_{n,\a_n<a}|\beta_n|e^{-\alpha_n\re(s)}\leq C,
\ee
for some $C>0$. Therefore $|\sum_{n,\a_n<a}\beta_n|\leq C e^{a\re(s)}$. The inverse Laplace transform of $s^{-1}\sum_{n=1}^{\infty}\beta_{n}e^{-\alpha_{n}s}$ is uniquely determined by the Lerch's theorem. Eq.\eqref{invLaplace} is the formula of the inverse Laplace transform.

\end{proof}

In applying \eqref{invLaplace} to the integrand of $\cn_{\leq}(a)$, we identify $\sum_{n=1}^{\infty}$ to $\sum_{N=1}^\infty\sum_{k_1,\cdots,k_n\in\Z_+}$ and
\be
\b_n\equiv \b(k),\qquad \a_n\equiv \sum_{i=1}^N\sqrt{k_i(k_i+2)}.
\ee
Assuming the assumption of Lemma \ref{lemmaV1} is satisfied by $\b(k)$, we obtain
\be
\cn_{\leq}(a)&=&\int_\cc \rmd\mu(z)\frac{1}{2\pi i}\int_{T-i\infty}^{T+i\infty}\frac{\rmd s}{s}e^{as}\left[\sum_{N=1}^\infty\sum_{k_1,\cdots,k_n\in\Z_+}\prod_{i=1}^N f(k_i,z)e^{-s\sum_{i=1}^N\sqrt{k_i(k_i+2)}}\right]\nonumber\\
&=&\int_\cc \rmd\mu(z)\frac{1}{2\pi i}\int_{T-i\infty}^{T+i\infty}\frac{\rmd s}{s}e^{as}\left[\sum_{N=1}^\infty\lt(\sum_{k\in\Z_+} f(k,z)e^{-s\sqrt{k(k+2)}}\rt)^N\right]\nonumber\\
&=&\int_\cc \rmd\mu(z)\,\frac{1}{2\pi i}\int_{T-i\infty}^{T+i\infty}\frac{\rmd s}{s}e^{as}\left(1-\sum_{k\in\Z_+}f\left(k,z\right)e^{-s\sqrt{k(k+2)}}\right)^{-1}.\label{cnleq}
\ee

For any $z\in\cc$, if the solutions to
\be
1-\sum_{k\in\Z_+}f\left(k,z\right)e^{-s\sqrt{k(k+2)}}=0
\ee 
give a finite number of simple poles $\{s_I(z)\}_I$ in the complex $s$-plane with $\re(s_I(z))>0$. By the residue theorem and $a\gg1$
\be
\cn_{\leq}\left(a\right)=\sum_I\int_\cc \rmd\mu(z)\, e^{a s_I(z)} s_I(z)^{-1}r_I(z)+\text{subleading },
\ee
where $r_I(z)$ is the residue. For $a\gg1 $, the integral over $\cc$ can be computed by the stationary phase method: We solve $\partial_z s_I(z)=0$ for the critical points $z_0$ and denote by $s_0=\max_{I,z_0} (\re [s_I(z_0)])$, we obtain
\be
\cn_{\leq}\left(a\right)= C a^{n} e^{a s_0}\lt[1+O(a^{-1})\rt]=\exp\lt[s_0 a+n\log(a)+O(1)\rt],
\ee
where $C$ is a constant. The coefficient $n$ of the logarithmic correction is generally negative, and $-2n\geq 0$ is the dimension of the Hessian matrix that is assumed to be nondegenerate.

The number of states in $[a-\delta,a+\delta]$ with $a\gg1$ and $\delta\sim O(1)$ is given by
\be
\cn_{\leq}\left(a+\delta\right)-\cn_{\leq}\left(a-\delta\right)&=&C a^{n} e^{a s_0} (e^{ s_0\delta }-e^{- s_0\delta })\lt[1+O(a^{-1})\rt].\label{a+da-d}
\ee
where $(e^{ s_0\delta }-e^{- s_0\delta })$ only give a correction to $O(1)$ on the exponent. As a result, we obtain in general the entropy
\be
S&=&\log\lt[\cn_{\leq}\left(a+\delta\right)-\cn_{\leq}\left(a-\delta\right)\rt]=s_0 a+n\log(a)+O(1)\nonumber\\
&=&\frac{2s_0}{8\pi\g\ell_P^2}\Ar(\Fs)+n\log\lt(\Ar(\Fs)\rt)+O(1).
\ee
This result is obtained from the general ansatz of $P(k)$ in \eqref{generalansatz}. Different choices of $\cc$, $\rmd\mu(z)$, and $f(k,z)$ only relate to two different parameters $s_0$ and $n$ in the entropy.

The entanglement entropy discussed in the last section corresponds to $P(s)=1$, namely $f(k,z)=1$ and trivial $\cc$. Then $s_0$ is trivially a constant:
\be
\sum_{k\neq 0}e^{-s_{0}\sqrt{\left(k+1\right)^{2}-1}}=1,\qquad 2s_0=\b_0\simeq 1.01527\ .
\ee
The $z$-integral is trivial, so there is no logarithmic correction: $n=0$. We reproduce the result in \eqref{resultS}:
\be
S(\bm{\rho}_c,A)=\log\cd = \frac{\beta_0}{8\pi\ell_{P}^{2}\gamma}\mathrm{Ar}(\mathfrak{S})+O(1).
\ee

\section{Bulk entropy}\label{Bulk entropy}

In the above discussion, we have assumed that the state $\Psi_{\Ar(\Fs)}$ does not have any entanglement between the bulk degrees of freedom inside $A$ and $\bA$. In terms of spin-networks, these bulk degrees of freedom include spins $\{j_\fl\}_{\fl\neq\fl_0}$, intertwiners, and graphs inside $A$ and $\bA$.

The bulk entanglement is turned on when the state in the bracket of \eqref{Psisum1} is not a pure tensor product, i.e. we relax the condition 2 in Section \ref{Fixed-area state} and consider the following family of states
\be
\Psi_{\Ar(\Fs)}=\sum_{\alpha=(N,\cj)}p_\alpha^{1/2}  \psi^{(\a)}_{A,\bA}\,\chi_{[-\Delta,\Delta]}\lt(C_1(\a)\rt),\qquad \psi^{(\a)}_{A,\bA}=\sum_{k,l}\psi^{(\a)}_{kl} \lt(e^k_{A,\alpha}\otimes\tilde{e}^l_{\bA,\alpha}\rt)\label{Psisum222}
\ee
where $e^k_{A,\alpha}$ and $\tilde{e}^l_{\bA,\alpha}$ form the orthonormal basis in $\hat\ch_{A,\a}$ and $\hat\ch_{\bA,\a}$ respectively. We assume $\psi_{\a;A,\bA}$ is normalized $\sum_{kl}\psi_{kl}^{(\a)}\bar{\psi}_{kl}^{(\a)}=1$ for each $\a$. In this case, the diagonal block given by $\bm{\rho}=|\Psi_{\Ar(\Fs)}\rangle \langle \Psi_{\Ar(\Fs)}|$ is given by
\be
\Tr_{\bA,\alpha}(\bm{\rho}_{\a\a})=p_\alpha \bm{\rho}_{A,\alpha},\qquad \bm{\rho}_{A,\alpha}=\sum_{k,k',l}\psi_{kl}^{(\a)}\bar{\psi}^{(\a)}_{k'l}|e^k_{A,\alpha}\rangle\langle e^{k'}_{A,\alpha}|.
\ee
In the entanglement entropy $S(\bm{\rho},A) $ in \eqref{entropyformula}, the reduced density matrix $\rho_{A,\a}$ has a nontrivial von Neumann entropy $S(\bm{\rho}_{A,\a})$ for generic $\psi_{kl}$
\be
S(\bm{\rho}_{A,\a})=-\Tr\lt(\psi^{(\a)}{\psi}^{(\a)\dagger}\log \psi^{(\a)}{\psi}^{(\a)\dagger}\rt),\label{totalS}
\ee
where the matrix $\psi^{(\a)}{\psi}^{(\a)\dagger}$ has elements $\sum_{l}\psi_{kl}^{(\a)}\bar{\psi}^{(\a)}_{k'l}$. The entropy $S(\bm{\rho}_{A,\a})$ describes the entanglement of the bulk degrees of freedom in $\hat{\ch}_{A,\a}\otimes\hat{\ch}_{\bA,\a}$, so we call $S(\bm{\rho}_{A,\a})$ is the bulk entropy.

\subsection{Random states}


Although here $\psi^{(\a)}_{A,\bA}$ with fixed $\a$ generally sums different internal spins on $\fl\neq \fl_0$, let us first restrict our attention to $\psi_{\a;A,\bA}$ with fixed internal spins. In this case, the Hilbert space is a finite tensor product of the intertwiner spaces, which are finite dimensional. When $|\psi_{\a;A,\bA}\rangle$ is a random state $|\psi_{\a;A,\bA}\rangle=U |\psi^{(0)}_{\a;A,\bA}\rangle$ where $\psi^{(0)}_{\a;A,\bA}$ is an arbitrary reference state and the unitary transformation $U$ is averaged over all unitary transformations (with Haar integration) on the Hilbert subspace with fixed $j_\fl$ for all $\fl\subset \G$, the random averaged bulk entropy $\langle S(\bm{\rho}_{A,\a})\rangle$ behaves as a quantum-information-theoretic volume law \cite{Bianchi:2023avf,Bianchi:2021aui}
\be
\langle S(\bm{\rho}_{A,\a})\rangle \sim \cn_A \overline{\log d_\fn}, \label{Nlogd}
\ee
for large graph $\G$. The random averaged entropy describes the typical behavior of the states in the Hilbert space. $\cn_A$ is the number of nodes in $A$ and is assumed to be less than half of the total number $\cn$ of nodes in $\G$. ${d_\fn}$ is the dimension of the intertwiner space $\ch_\fn(\{j_\fl\})$, and $\overline{\cdots}$ is the average over nodes in $A$. This result is due to the fact that the random state is maximally entangled approximately for large dimension. The large dimension also makes the dispersion around the averaged value exponentially small \cite{Page:1993df}. The right-hand side of \eqref{Nlogd} is the logarithmic of the dimension of $\otimes_{\fn\in A}\ch_\fn(\{j_\fl\})$ for fixed $j_\fl$, and the dimension is large for large $\cn_A$ and $\cn$. We emphasize that the volume law in \eqref{Nlogd} is quantum-information-theoretic, because \eqref{Nlogd} relates to $\cn_A$ rather than the geometrical volume measured by the LQG volume operator.

This argument can be generalized to random states in the larger Hilbert space $\hat{\ch}_{A,\a}\otimes\hat{\ch}_{\bA,\a}$, where the internal spins $j_\fl$, $\fl\neq \fl_0$ are not fixed. $\hat{\ch}_{A,\a},\hat{\ch}_{\bA,\a}$ are arbitrary finite-dimensional truncations of ${\ch}_{A,\a},{\ch}_{\bA,\a}$. We consider $|\psi_{\a;A,\bA}\rangle$ to be a random state $|\psi_{\a;A,\bA}\rangle=\mathcal{U}_{\a} |\psi^{(0)}_{\a;A,\bA}\rangle$, where $\psi^{(0)}_{\a;A,\bA}\in\hat{\ch}_{A,\a}\otimes\hat\ch_{\bA,\a}$ is an arbitrary reference state and the unitary transformation $\mathcal{U}_{\a}$ is averaged over all unitary transformations (with Haar integration) on $\hat\ch_{A,\a}\otimes\hat\ch_{\bA,\a}$. We assume both $ \dim(\hat\ch_{A,\a}) $ and $ \dim(\hat\ch_{\bA,\a})$ are large. Then the random averaged bulk entropy $\langle S(\bm{\rho}_{A,\a})\rangle$ is given approximately by
\be
\langle S(\bm{\rho}_{A,\a})\rangle \simeq \log\lt[ {\rm min}\lt\{\dim(\hat\ch_{A,\a}),\dim(\hat\ch_{\bA,\a})\rt\}\rt]. \label{logdim}
\ee
The dispersion around this averaged value are suppressed by the large dimension. Some details of deriving \eqref{logdim} are given in Appendix \ref{random average}.

\subsection{Renormalization of area-law}\label{Renormalization of area-law}

Given the nontrivial bulk entropy $S(\bm{\rho}_{A,\a})$, the maximization of $S(\bm\rho, A)$ in \eqref{maximization1} needs to be modified to include $S(\bm{\rho}_{A,\a})$
\be
\delta S\left(\bm\rho,A\right)=-\sum_{\alpha}\delta p_{\alpha}\log p_{\alpha}+\left(\sigma-1\right)\sum_{\alpha}\delta p_{\alpha}+\sum_{\alpha}\delta p_{\alpha}S\left(\bm{\rho}_{A,\a}\right)=0,
\ee
which implies the non-constant $p_\alpha$ 
\be
p_{\alpha}=\frac{e^{S\left(\bm{\rho}_{A,\a}\right)}}{\sum_{\alpha}e^{S\left(\bm{\rho}_{A,\a}\right)}}.
\ee
The maximized entanglement entropy $S\left(\bm\rho_{c},A\right)$ is given by
\be
S\left(\bm\rho_{c},A\right)=\log\sum_{\alpha}e^{S\left(\bm{\rho}_{A,\a}\right)}\chi_{[-\Delta,\Delta]}\lt(C_1(\a)\rt)\ .\label{SlogsumeS}
\ee
In the case of small bulk entanglement $S(\bm{\rho}_{A,\a})\ll1 $, we have the appoximation
\be
S\left(\bm\rho_{c},A\right)\simeq \log{\cal D}+S_{\rm bulk}\simeq \frac{\b_0}{8\pi\g\ell_P^2} \Ar(\Fs)+S_{\rm bulk}, \label{perturbSbulk}
\ee
where $S_{\rm bulk}=\cd^{-1}\sum_\alpha S(\bm{\rho}_{A,\a})$ is the averaged bulk entropy over $\a=(N,\cj)$. This formula shares some similarity with the generalized Bekenstein-Hawking entropy, which contains both the area-law entropy and the entropy of the bulk degrees of freedom.

$e^{S\left(\bm{\rho}_{A,\a}\right)}$ can be viewed effectively as the degeneracy at each $\a=(N,\cj)$, i.e. $e^{S\left(\bm{\rho}_{A,\a}\right)}$ is the same as $P(k)$ in Section \ref{A general framework of state-counting}. Then $\sum_{\alpha}e^{S\left(\bm{\rho}_{A,\a}\right)}$ in \eqref{SlogsumeS} can still be viewed as as counting the number of microstates. As an example, let us consider $e^{S\left(\bm{\rho}_{A,\a}\right)}$ as a function of $\{j_{\fl_0}\}$ to be factorized into punctures,
\be
e^{S\left(\bm{\rho}_{A,\a}\right)}=\prod_{\fl_0=1}^N G(j_{\fl_0}).\label{factorbulk}
\ee
This can be obtained by the following factorized state
\be
\psi_{A,\bar{A}}^{(\alpha)}=\bigotimes_{\mathfrak{l}_{0}}F_{A,\bA}(j_{\fl_0}),\qquad F_{A,\bA}(j_{\fl_0})=\sum_{a,b}F_{ab}\left(j_{\mathfrak{l}_{0}}\right)\left(e_{A,j_{\mathfrak{l}_{0}}}^{a}\otimes\tilde{e}_{\bar{A},j_{\mathfrak{l}_{0}}}^{b}\right).\label{ansatzpsi}
\ee
The state $F_{A,\bA}(j_{\fl_0})$ is defined on a subgraph $\G_{\fl_0}\subset\G_0$ intersecting $\Fs$ only with a single link $\fl_0$ with $j_{\fl_0}\neq 0$. $\{e_{A,j_{\mathfrak{l}_{0}}}^{a}\}_a ,\{\tilde{e}_{\bar{A},j_{\mathfrak{l}_{0}}}^{b}\}_b$ are orthonormal. For the graph on which $\psi_{A,\bar{A}}^{(\alpha)}$ in \eqref{ansatzpsi} is based, each node close to $\Fs$ has only one link intersect $\Fs$. The matrix $\psi^{(\alpha)}$ in \eqref{totalS} is factorized $\psi^{(\alpha)}=\otimes_{\mathfrak{l}_{0}}F(j_{\fl_0})$. We assume that the matrix $F$ depends on $\fl_0$ only through $j_{\fl_0}$, then
\be
G(j_{\fl_0})=e^{-\Tr\lt(F(j_{\fl_0})F(j_{\fl_0})^\dagger\log F(j_{\fl_0})F(j_{\fl_0})^\dagger\rt)}
\ee
The difference between $S(\bm{\rho_c},A)$ in \eqref{SlogsumeS} and $\log \cd$ is only given by including the degeneracy $G(j)\neq 1$ in \eqref{Dstatecount}, i.e.
\be
\sum_{\a}e^{S\left(\bm{\rho}_{A,\a}\right)}=\tilde{\sum_{\{n_j\}}} d[\{n_j\}],\quad d[\{n_j\}]=\lt(\sum_{j\neq 0} n_j \rt)!\, \prod_{j\neq0}\frac{G(j)^{n_j}}{n_j!}
\ee
The degeneracy modifies $\bar{n}_j$ in \eqref{boltzmann} and the normalization by
\be
\frac{\bar{n}_j}{\sum_{j\neq 0}\bar{n}_j}=e^{-\b\sqrt{j(j+1)}}G(j),\qquad \sum_{j\neq 0}e^{-\b\sqrt{j(j+1)}}G(j)=1.\label{betaGj}
\ee
In the method of Section \ref{A general framework of state-counting}, the modification corresponds to change from $P(k)=1$ to 
\be
P(k)=\prod_{k\neq 0}G(k)^{n_k}, \qquad\text{or}\qquad f(k,z)=G(k),
\ee
while the $z$-integral is still trivial. The resulting $S\left(\bm\rho_{c},A\right) $ still gives the area law at the leading order and no logarithmic correction
\be
S(\bm{\rho}_c,A) =  \frac{\b}{8\pi\g \ell_P^2}\Ar(\Fs)+O(1),
\ee
where $\beta$ determined by \eqref{betaGj} is different from $\b_0$ for $G(j)\neq 1$. In this example, the bulk entanglement effectively ``renormalizes'' the coefficient of the area-law:
\be
\frac{\b_0}{8\pi\g \ell_P^2}\to \frac{\b}{8\pi\g \ell_P^2}.
\ee

The renormalization of the coefficient can also be equivalently understood as the renormalization of Newton's constant: If we identify
\be
\frac{\b_0}{8\pi\g \ell_P^2}=\frac{1}{4[\ell_P^2]_{0}},\qquad \frac{\b}{8\pi\g \ell_P^2}=\frac{1}{4[\ell_P^2]_{\rm ren}},
\ee
then we obtain the the renormalization
\be
[\ell_P^2]_{\rm ren}=\frac{\b_0}{\b}[\ell_P^2]_0.
\ee 
This is an analog of the known result for black hole entropy: Given that the generalized Bekenstein-Hawking entropy that contains both the area-term and the bulk entropy of quantum fields, the bulk entropy is divergent due to the vacuum entanglement entropy of short-wavelength modes across the horizon. The leading-order divergence of the bulk entropy is proportional to the area and thus corresponds to a renormalization of Newton's constant \cite{Larsen:1995ax,Kabat:1995eq,Jacobson:1994iw}. See also \cite{Ghosh:2012wq} for the renormalization in context of the LQG black hole.

\subsection{Logarithmic correction}\label{Logarithmic correction}

The bulk entropy can result in the logarithmic correction to the area-law. As an example, let us assume every graph $\G$ involved in $\ch$ has only a single node $\fn_A$ in $A$, and every $\psi^{(\a)}_{A,\bA}$ is either a maximal entangled state or a random state. Then $e^{S\left(\bm{\rho}_{A,\a}\right)}$ equals the dimension of the intertwiner space at the node $\fn_A$. All links connecting to the node $\fn_A$ make punctures on $\Fs$. Then
\be
P(k)=e^{S\left(\bm{\rho}_{A,\a}\right)}=\frac{1}{\pi}\int_{0}^{2\pi}d\theta\sin^{2}\left(\theta\right)\prod_{k}\left[\frac{\sin\left(\left(k+1\right)\theta\right)}{\sin\left(\theta\right)}\right]^{n_{k}}.\label{eq75}
\ee
Both the triangle inequality and the constraint that the sum of $k$ must be even are imposed by $P(k)$ to the state-sum $\cn_\leq(a)$ in \eqref{cn(a)}.

Comparing to the general ansatz \eqref{generalansatz}, we identify the auxiliary space $\cc$ to be a circle and 
\be
\int_\cc d\mu\left(z\right)\cdots =\frac{1}{\pi}\int_{0}^{2\pi}d\theta\sin^{2}\left(\theta\right)\cdots,\qquad f\left(k,z\right)=\frac{\sin\left(\left(k+1\right)\theta\right)}{\sin\left(\theta\right)}.
\ee
In this case, the entanglement entropy $S(\bm{\rho}_c,A) $ equals to the SU(2) black hole entropy computed in \cite{Engle2011,Agullo:2009eq}. Indeed, the formula \eqref{cnleq} of $\cn_{\leq}\left(a\right)$ becomes
\be
\cn_{\leq}\left(a\right)=\frac{1}{\pi}\int_{0}^{2\pi}d\theta\sin^{2}\left(\theta\right)\frac{1}{2\pi i}\int_{T-i\infty}^{T+i\infty}\frac{ds}{s}\left(1-\sum_{k\neq0}\frac{\sin\left(\left(k+1\right)\theta\right)}{\sin\left(\theta\right)}e^{-s\sqrt{\left(k+1\right)^{2}-1}}\right)^{-1}e^{as}\ .\nonumber
\ee
The pole with maximal real part is given by
\be
\sum_{k\neq 0}\left(k+1\right)e^{-s_{0}\sqrt{\left(k+1\right)^{2}-1}}=1,\quad\Rightarrow\quad \frac{s_0}{\pi}= 0.274067\cdots,\label{blackholeBetavalue}
\ee
$\beta=2s_0\neq \b_0$ renormalizes the coefficient of the area-law as in the previous example. For a small nonzero $\theta$, the deformation of pole is given by 
\be
s=s_{0}-\alpha\theta^{2}+O\left(\theta^{4}\right),
\ee
for a constant $\a>0$. $\cn_{\leq}\left(a\right)$ can be approximated by
\be
\cn_{\leq}\left(a\right)  
\simeq e^{s_{0}a}\frac{1}{\pi s_{0}}\int_{-\epsilon}^{\epsilon}d\theta e^{-a\alpha\theta^{2}}\theta^{2} =\exp\left[s_{0}a-\frac{3}{2}\log\left(a\right)+O\left(1\right)\right],\qquad \forall \epsilon >0.
\ee
The entanglement entropy gives the logarithmic correction in addition to the renormalization of the area-law:
\be
 S(\bm{\rho}_c,A) = \frac{2s_{0}}{8\pi\gamma\ell_{P}^{2}}\mathrm{Ar}(\Fs)-\frac{3}{2}\log\left(\mathrm{Ar}(\Fs)\right)+O\left(1\right).
\ee

\section{Von Neumann algebra and entanglement entropy: Infinite-dimensional Hilbert space}\label{VNA Infinite-dimensional Hilbert space}

\subsection{Von Neumann algebra of holonomies and fluxes}

We generalize the discussion of von Neumann algebra in Section \ref{finitedimcase} to infinite-dimensional Hilbert space. We focus on a single graph $\G_0$ with finitely many links. The Hilbert space $\ch=\ch_{\G_0}$ on the graph is separable. 

It turns out to be more convenient to enlarge $\ch$ to the Hilbert space including non-gauge-invariant states: $\ch^{\rm (aux)}\simeq L^{2}(\Su^{\times L})$, where $L$ is the number of links in $\G_0$, and define firstly the von Neumann algebra $\cm^{\rm (aux)}$ containing holonomies and (exponentiated) fluxes on $\ch^{\rm (aux)}$. Then $\cm$ is given by $\cm=\bm{p}^{\rm inv}\cm^{\rm (aux)}\bm{p}^{\rm inv}$, where $\bm{p}^{\rm inv}$ is the SU(2) gauge invariant projection from $\ch^{\rm (aux)}$ to $\ch$.

The Hilbert space $\ch^{\rm (aux)}\simeq \otimes_\fl \ch_{\fl}$ has been discussed in \eqref{htilde}, and it has the following decomposition according to $A$ and $\bA$ similar to \eqref{chPj}:
\be
\ch^{\rm (aux)}\simeq \bigoplus_{\alpha}\ch^{\rm (aux)}_{A,\alpha}\otimes \ch^{\rm (aux)}_{\bA,\alpha},
\ee
where $\ch_{A,\alpha }^{\rm (aux)}$ and $\ch_{\bA,\alpha }^{\rm (aux)}$ are expressed in terms of $\widetilde{\ch}_{\fn}$ in \eqref{htilde}:
\be
\ch_{A,\alpha }^{\rm (aux)}=\bigoplus_{\{j_\fl\}_{\fl\subset A}}\bigotimes_{\fn\in\G_0\cap A}\widetilde{\ch}_\fn(\{j_\fl\}),\qquad \ch_{\bA,\alpha }^{\rm (aux)}=\bigoplus_{\{j_\fl\}_{\fl\subset \bA}}\bigotimes_{\fn\in\G_0\cap \bA}\widetilde{\ch}_\fn(\{j_\fl\}).
\ee

Given a Hilbert space $\ch$ and $\cb(\ch)$ the $*$-algebra of bounded operators, a von Neumann algebra is a unital $*$-subalgebra of $\cb(\ch)$ and closed under weak operator topology. Here for any $\alpha$, we define the von Neumann algebra $\cm^{\rm (aux)}_\alpha$ by the algebra of bounded operators on $\ch_{A,\alpha }^{\rm (aux)}$:
\be
\cm^{\rm (aux)} _\alpha:= \cb\lt(\ch_{A,\alpha }^{\rm (aux)}\rt)\otimes \bm{I}^{\rm (aux)}_{\bA,\alpha}, \label{cbHaux}
\ee
where $\bm{I}^{\rm (aux)}_{\bA,\alpha}$ is the identity operator on $\ch_{\bA,\alpha }^{\rm (aux)}$. The von Neumann algebra $\cm^{\rm (aux)}$ on $\ch^{\rm (aux)}$ is the direct sum
\be
\cm^{\rm (aux)}=\bigoplus_\alpha\cm^{\rm (aux)}_\alpha.
\ee
It is clear that $\cm^{\rm (aux)}$ contains the following bounded operators on $\ch^{\rm (aux)}$:

\begin{itemize}

    \item All continuous function on $\Su^{\times L_A}$: The continuous functions $f(\bm{h}_{\fl})$ of holonomy operators $\bm{h}_{\fl}$, where $\fl\subset A$ do not intersect with $\Fs$ (the total number of these links is $L_A$)
    
    \item The Weyl operators of fluxes: $\bm{w}_{\vec{t}}=\exp\lt(i\sum_{\fn,\fl}\vec{t}_{(\fn,\fl)}\cdot \vec{\bm{J}}_{(\fn,\fl)}\rt)$, where $\vec{t}_{(\fn,\fl)}\in\R^3$, $\fn\in A$, and $\fl$ includes $\fl\subset A$ and $\fl_0$ intersecting $\Fs$. $ \vec{\bm{J}}_{(\fn,\fl)}=\vec{\bm{R}}_\fl$ (or $\vec{\bm{L}}_\fl$) acting on $V_{j_\fl}$ (or $V^*_{j_\fl}$) if $\fn$ is the target (or source) of $\fl$. Here all $V_{j_\fl}$ and $V^*_{j_\fl}$ acted by $\vec{\bm{J}}_{(\fn,\fl)}$ belong to $\ch_{A,\alpha }^{\rm (aux)}$.

\end{itemize}
However, if we only focus on the subalgebra formed by the continuous functions of holonomies and the Weyl operators of fluxes. This subalgebra by itself is a unital $*$-algebra, whose closure under weak operator topology defines a von Neumann subalgebra $\cm_0\subseteq \cm^{\rm (aux)}$. It turns out that $\cm_0$ is identical to $ \cm^{\rm (aux)}$, and the argument is given below.

\begin{lemma}
$\cm_0$ contains $\bm{x}_{\vec{s}}=\exp\lt(i\sum_{\fl} {s}_\fl {{\vec{\bm{J}}_{(\fn,\fl)}}\cdot \vec{\bm{J}}_{(\fn,\fl)}}\rt)$, where $s_{\fl}\in\R$, and $\fl$ includes $\fl\subset A$ and $\fl_0$ intersecting $\Fs$.
\end{lemma} 

\begin{proof} 

We skip the indices $(\fn,\fl)$ and denote by $\bm w_{t}^i=\exp(it \bm{J}^i)$. We define $\bm{J}^i(t)=\frac{1}{it}(\bm w_{t}^i-1)$. Both $\bm{J}^i(t)$ and $\bm{J}^i(t)^\dagger$ belongs to $\cm_0$. By the operator topology on $\cm_0$, the unitary operator $\exp[is\sum_i\bm{J}^i(t)^\dagger \bm{J}^i(t)]$ also belongs to $\cm_0$ \footnote{Given any bounded operator $\bm{B}$, the exponential $\exp(\bm B)=\lim_{N\to\infty}\sum_{n=0}^N\frac{\bm B^n}{n!}$ converges in operator norm (so it converges in both strong and weak topologies): $\|\sum_{n=0}^N\frac{\bm B^n}{n!}-\sum_{n=0}^M\frac{\bm B^n}{n!}\|\leq \sum_{n=N+1}^M\frac{\|\bm B\|^n}{n!}\to 0$ (and every Cauchy sequence converges in the space of bounded operators).}. On the dense domain $\Fd^{\rm (aux)}$ spanned by finite linear combinations of spin-network states, we have\footnote{The spin-network states with different $\mathbf{j}=\{j_\fl\}_{\fl\subset \G_0}$ are orthogonal. The subspace of spin-networks with fixed $\mathbf{j}$ is finite dimensional, and $\bm w_t^i$ and $\bm{J}^i$ leave this subspace invariant. Therefore, Eq.\eqref{weaklimitondense} reduces to the limits of the finite-dimensional matrices.}
\be
\lim_{t\to 0}\lag \psi \lt|e^{is\sum_i\bm{J}^i(t)^\dagger \bm{J}^i(t)}\rt|\phi \rag =\lag \psi \lt|e^{is\sum_i\bm{J}^i\bm{J}^i}\rt|\phi \rag,\qquad \forall \psi,\phi\in\Fd^{\rm (aux)}.\label{weaklimitondense}
\ee
Since $\|\exp[is\sum_i\bm{J}^i(t)^\dagger \bm{J}^i(t)] \|=1$ for all $t,s\in\R$, the above weak convergence can be extended to the entire $\ch^{\rm (aux)}$ \footnote{Any $\psi,\phi\in\ch^{\rm (aux)}$ can be approached by the limits of the sequences $\{\psi_i\}_{i=1}^\infty,\{\phi_j\}_{j=1}^\infty$ where $\psi_i,\phi_j\in\Fd^{\rm (aux)}$. Denote by $\bm u_s(t)=\exp[is\sum_i\bm{J}^i(t)^\dagger \bm{J}^i(t)]$ and $\bm u_s(0)=\exp[is\sum_i\bm{J}^i \bm{J}^i]$. $|\langle \psi |\bm u_s(t)|\phi \rangle-\langle \psi_i |\bm u_s(t)|\phi_j \rangle|\leq |\langle \psi |\bm u_s(t)|\phi -\phi_j \rangle|+|\langle \psi - \psi_i |\bm u_s(t)|\phi_j \rangle|\leq \|\psi\|\|\phi -\phi_j\|+\|\psi - \psi_i\|\|\phi_j\|\leq \eps$ and similarly $|\langle \psi |\bm u_s(0)|\phi \rangle-\langle \psi_i |\bm u_s(0)|\phi_j \rangle|\leq \eps$ for sufficiently large $i,j$. Therefore, $|\langle \psi |\bm u_s(t)|\phi \rangle-\langle \psi |\bm u_s(0)|\phi \rangle|\leq |\langle \psi |\bm u_s(t)|\phi \rangle-\langle \psi_i |\bm u_s(t)|\phi_j \rangle|+|\langle \psi_i |\bm u_s(t)|\phi_j \rangle-\langle \psi_i |\bm u_s(0)|\phi_j \rangle|+|\langle \psi |\bm u_s(0)|\phi \rangle-\langle \psi_i |\bm u_s(0)|\phi_j \rangle|\leq 3\eps$ for $t$ close to 0.}.

\end{proof}

\begin{lemma}
    $\cm_0$ is identical to $\cm^{\rm (aux)}$.
\end{lemma}

\begin{proof}

The operators $\bm{x}_{\vec{s}}$ and $\exp\lt(i\sum_{\fn,\fl}{t}^3_{(\fn,\fl)} {\bm{J}}^3_{(\fn,\fl)}\rt)$ form an abelian $*$-subalgebra and gives the spectral projections $\bm{P}_{(\mathbf{j,m,n})}$, each of which projects one of $\ch_{A,\alpha }^{\rm (aux)}$ to 1-dimensional subspace
\be
\bm{P}_{(\mathbf{j,m,n})}\ch^{\rm (aux)}=\bbc|\mathbf{j,m,n}\rangle\otimes \ch_{\bA,\alpha }^{\rm (aux)},\qquad \bm{P}_{(\mathbf{j,m,n})}=|\mathbf{j,m,n}\rangle\langle\mathbf{j,m,n}|\otimes \bm{I}^{\rm (aux)}_{\bA,\alpha},\label{minimalproj}
\ee
where $|\mathbf{j,m,n}\rangle$ denotes the tensor product state whose factors are $|j_\fl,m_\fl\rangle\in V_{j_\fl}$ and $\langle j_{\fl'},n_{\fl'}|\in V^*_{j_{\fl'}}$ for $ V_{j_\fl},V^*_{j_{\fl'}}$ belongs to $\ch_{A,\alpha }^{\rm (aux)}$. All $\bm{P}_{(\mathbf{j,m,n})}$ belong to $\cm_0$ as the minimal projections. The states $|\mathbf{j,m,n}\rangle$ form an orthonormal basis in $\ch^{\rm (aux)}_{A,\alpha}$. 

The operators $\bm w_{\vec{t}}$ acting on $|\mathbf{j,m,n}\rangle$ leaves $\mathbf{j}$ invariant but can give linear combinations of $|\mathbf{j},\bf{m}',\bf{n}'\rangle$ with different $\bf{m}',\bf{n}'$. On the other hand, a generic $f(\bm{h}_\fl)$ acting on $|\mathbf{j,m,n}\rangle$ can give linear combinations of $|\mathbf{j}',\bf{m}',\bf{n}'\rangle$ with different $\bf{j}'$ (and different $\bf{m}',\bf{n}'$) but leaves $\alpha$ invariant. Therefore the actions of $(\bm{w}_{\vec{t}}f(\bm{h}_\fl)\otimes \bm{I}^{\rm (aux)}_{\bA,\alpha})\bm{P}_{(\mathbf{j},\bf{m},\bf{n})}$ followed by the projection $\bm{P}_{(\mathbf{j}',\bf{m}',\bf{n}')}$ and a rescaling equals to the partial isometry
\be
\bm{V}^{(\mathbf{j}',\mathbf{m}',\mathbf{n}')}_{\ (\mathbf{j},\mathbf{m},\mathbf{n})}=|\mathbf{j}',\mathbf{m}',\mathbf{n}'\rangle\langle \mathbf{j},\mathbf{m},\mathbf{n}|\otimes \bm{I}^{\rm (aux)}_{\bA,\alpha}
\ee
The partial isometry $\bm{V}^{(\mathbf{j}',\bf{m}',\bf{n}')}_{\ (\mathbf{j},\bf{m},\bf{n})}$ for any pair $(\mathbf{j}',\bf{m}',\bf{n}')$ and $(\mathbf{j},\bf{m},\bf{n})$ with fixed $\alpha$ can be obtained in this way. It belongs to $\cm_0$ since all ingredients belong to $\cm_0$ and they satisfy
\be
\bm{P}_{(\mathbf{j}',\mathbf{m}',\mathbf{n}')}=\bm{V}^{(\mathbf{j}',\mathbf{m}',\mathbf{n}')}_{\ (\mathbf{j},\mathbf{m},\mathbf{n})}\bm{V}^{(\mathbf{j}',\mathbf{m}',\mathbf{n}')}_{\ (\mathbf{j},\mathbf{m},\bf{n})}{}^\dagger,\qquad \bm{P}_{(\mathbf{j},\bf{m},\bf{n})}=\bm{V}^{(\mathbf{j}',\bf{m}',\bf{n}')}_{\ (\mathbf{j},\bf{m},\bf{n})}{}^\dagger \bm{V}^{(\mathbf{j}',\bf{m}',\bf{n}')}_{\ (\mathbf{j},\bf{m},\bf{n})}.
\ee
It means that all projections $\bm{P}_{(\mathbf{j},\bf{m},\bf{n})}$ are equivalent. Furthermore, any operator $\oplus_\alpha(\bm{O}_\alpha\otimes\bm{I}^{\rm (aux)}_{\bA,\alpha} )\in \cm^{\rm (aux)}$ has the expansion\footnote{The proof of this equality: the operators on the left and right hand sides have the same matrix elements on the spin-network states. The operator on the left hand side is bounded. Given a bounded operator $A$ and another operator $B$ on some Hilbert space $\ch$, and given a orthonormal basis $\{e_i\}$ in $\ch$, $\langle e_i|A|e_j\rangle=\langle e_i|B|e_j\rangle$ implies $A e_i=B e_i$ for all $i$, i.e. $A=B$ on a dense domain $D$ spanned by the basis. The bounded operator $A$ is defined on entire $\ch$, so $A=B$ implies that $B$ is at least densely defined on $D$. The B.L.T theorem indicates that $B$ has the unique bounded extension $\bar{B}$ and $A=\bar{B}$ on $\ch$.}
\be
\bigoplus_\alpha\lt(\bm{O}_\alpha\otimes\bm{I}^{\rm (aux)}_{\bA,\alpha} \rt)&=&\bigoplus_\alpha\sum _{(\mathbf{j},\mathbf{m},\mathbf{n})}\sum _{(\mathbf{j}',\mathbf{m}',\mathbf{n}')}|\mathbf{j}',\mathbf{m}',\mathbf{n}'\rangle \lt(O_\alpha\rt)_{(\mathbf{j}',\mathbf{m}',\mathbf{n}')}^{(\mathbf{j},\mathbf{m},\mathbf{n})}\langle \mathbf{j},\mathbf{m},\mathbf{n}|\otimes \bm{I}^{\rm (aux)}_{\bA,\alpha}\nonumber\\
&=&\bigoplus_\alpha\sum _{(\mathbf{j},\mathbf{m},\mathbf{n})}\sum _{(\mathbf{j}',\mathbf{m}',\mathbf{n}')}\bm{V}^{(\mathbf{j}',\mathbf{m}',\mathbf{n}')}_{\ (\mathbf{j},\mathbf{m},\mathbf{n})}\lt(O_\alpha\rt)_{(\mathbf{j}',\mathbf{m}',\mathbf{n}')}^{(\mathbf{j},\mathbf{m},\mathbf{n})},
\ee
where the possibly infinite sum on the right hand side is under the weak limit. Then $\oplus_\alpha (\bm{O}_\alpha\otimes\bm{I}^{\rm (aux)}_{\bA,\alpha} )\in \cm_0$ since all the partial isometries belongs to $\cm_0$. This shows $\cm^{\rm (aux)}\subset \cm_0$. Since $\cm_0\subseteq \cm^{\rm (aux)}$ by definition, we obtain $\cm_0 = \cm^{\rm (aux)}$.

\end{proof}

Let us remind some terminology of projectors in von Neumann algebra: Given two projectors $\bm{P},\bm{Q}$ in a von Neumann algebra $\cm$, if there exists a partial isometry $\bm{V}\in\cm$ such that $\bm{P}=\bm{V}^\dagger \bm{V}$ and $ \bm{Q}= \bm{V} \bm{V}^\dagger$, we say $\bm{P}$ and $\bm{Q}$ are equivalent: $\bm{P}\sim \bm{Q}$. A projection $\bm{P}\in\cm$ is finite, if every proper sub-projection $\bm{P}'\in \cm$, $\bm{P}' < \bm{P}$ is inequivalent to $\bm{P}$. $\bm{P}$ is infinite if there exists a proper sub-projection that is equivalent to $\bm{P}$. A finite projection is minimal if there is no proper sub-projection except zero.

A factor is type I if it contains a nonzero minimal projection operator. The von Neumann algebra $\cm^{\rm (aux)}_\alpha$ on $\ch^{\rm (aux)}_{A,\alpha}\otimes \ch^{\rm (aux)}_{\bA,\alpha}$ is a type I factor. The definition \eqref{cbHaux} is the standard form a type I factor \cite{Sorce:2023fdx}.  In $\cm^{\rm (aux)}_\alpha$, all $\bm{P}_{(\mathbf{j},\mathbf{m},\mathbf{n})}$ are minimal projections. There are infinitely many mutually orthogonal minimal projections, so $\cm^{\rm (aux)}_\alpha=\cm_0$ is of type $\mathrm{I}_\infty$.

We denote by $\bm{p}^{\rm inv}$ the projection from $\ch^{\rm (aux)}_\G$ to $\ch$ of gauge invariant states. $\bm{p}^{\rm inv}$ leaves $\alpha$ invariant and satisfies $\bm{p}^{\rm inv}=\bm{p}^{\rm inv}_A\otimes \bm{p}^{\rm inv}_{\bA}$ where $\bm{p}^{\rm inv}_{A}:\ch^{\rm (aux)}_{A,\alpha}\to\ch_{A,\alpha}$ and $\bm{p}^{\rm inv}_{\bA}:\ch^{\rm (aux)}_{\bA,\alpha}\to\ch_{\bA,\alpha}$. We define the von Neumann algebra $\cm$ of gauge invariant operators on $\ch$ by the projection of the algebra $\cm_0$ of the holonomies and fluxes:
\be
\cm:=\bm{p}^{\rm inv}\cm_0\bm{p}^{\rm inv}.
\ee

\begin{theorem}
$\cm$ is a type $I$ non-factor von Neumann algebra, and $\cm\simeq \oplus_\a\cb(\ch_{A,\alpha})\otimes \bm{I}_{\bA,\alpha}$.
\end{theorem}

\begin{proof}

This follows from $\cm_0=\cm^{\rm(aux)}=\oplus_\alpha\lt(\cb(\ch^{\rm (aux)}_{A,\alpha})\otimes \bm{I}^{\rm (aux)}_{\bA,\alpha}\rt)$ and $\bm{p}^{\rm inv}_A\cb(\ch^{\rm (aux)}_{A,\alpha})\bm{p}^{\rm inv}_A =\cb(\ch_{A,\alpha})$ 
\footnote{
Obviously $\bm{p}^{\rm inv}_A\cb(\ch^{\rm (aux)}_{A,\alpha})\bm{p}^{\rm inv}_A \subseteq \cb(\ch_{A,\alpha})$. Moreover, given that $\ch^{\rm (aux)}_{A,\alpha}=\oplus_{I}\ch_{A,\alpha}^{(I)}$ where each $I$ is a collection of total angular momentum at all $\fn\in A$ and $\ch_{A,\alpha}=\ch_{A,\alpha}^{(0)}$, consider a subalgebra $\oplus_I\cb(\ch_{A,\alpha}^{(I)})$ in $\cb(\ch^{\rm (aux)}_{A,\alpha})$, any $\bm\co\in \cb(\ch_{A,\alpha})=\cb(\ch^{(0)}_{A,\alpha})$ can be written as $\bm\co=\bm{p}^{\rm inv}_A(\oplus_I \bm\co_I)\bm{p}^{\rm inv}_A$, where $\oplus_I \bm\co_I\in \oplus_I\cb(\ch_{A,\alpha}^{(I)})\subset \cb(\ch^{\rm (aux)}_{A,\alpha})$ and $\bm\co_0=\bm\co$, $\bm\co_{I\neq 0}=\bm I$, so $\cb(\ch_{A,\alpha})\subseteq \bm{p}^{\rm inv}_A\cb(\ch^{\rm (aux)}_{A,\alpha})\bm{p}^{\rm inv}_A $.
}.

\end{proof}

\subsection{Renormalized trace, density operator, and entanglement entropy}

Given any von Neumann algebra $\cm$ and its subset of positive operators $\cm_+$, a trace on $\cm$ is a map $\hat{\Tr}: \cm_+\to [0,\infty]$ satisfying (1) $\hat{\Tr}(\l \bm T) =\l \hat{\Tr}(\bm T)$ and $\hat{\Tr}(\bm T+\bm S)=\hat{\Tr}(\bm T)+\hat{\Tr}(\bm S)$ for all $\bm T,\bm S\in\cm_+$ and all $\l\geq 0$, (2) $\hat{\Tr}(\bm U \bm T \bm U^*)=\hat{\Tr}(\bm T)$ for all $\bm T\in\cm_+$ and all unitary $\bm U\in\cm$. These defining properties implies the trace satisfies the cyclic property $\hat{\Tr}(\bm S \bm T)=\hat{\Tr}(\bm T \bm S)$ when extending to trace-class operators. The trace $\hat{\Tr}$ is renormalized if it satisfy (1) $\hat{\Tr}(\bm T)=0$ implies $\bm T=0$, (2) $\hat{\Tr}(\bm P)$ is finite for every finite projector $\bm P$, (3) If $\{\bm \rho_I\}$ is a family of positive operators in $\cm_+$ with a supremum $\bm \rho = \mathrm{sup}_I \bm \rho_I$, then we
have $\hat{\Tr}(\bm \rho) = \mathrm{sup}_I\hat{\Tr}(\bm \rho_I)$. Every type I von Neumann factor has a unique renormalized trace up to rescaling $\hat{\Tr}\to \l \hat{\Tr}$, $\forall\l>0$. The renormalized trace is determined by setting every minimal projector $\bm{P}$ to have $\hat{\Tr}(\bm P)=\l$. See \cite{Sorce:2023fdx} for a detailed discussion on the renormalized trace and the uniqueness.

In our case of $\cm=\oplus_\a \cm_\a$ where $\cm_\a$ are type-I factors, for $\bm \co=\oplus_\a\bm \co_\a\in\cm,\  \bm \co_\a=\bm \co_{A,\a}\otimes \bm I_{\bA,\a}\in\cm_\a$ such that $\bm \co_{A,\a}$ is in trace-class, we define the renormalized trace is by
\be
\hat{\Tr}(\bm \co)=\sum_\a \hat{\Tr}_\a(\bm \co_\a), \quad \text{where}\quad \hat{\Tr}_\a(\bm \co_\a)=\l_\a \Tr_{A,\a} (\bm \co_{A,\a}),\quad \l_\a>0. \label{hatTr84}
\ee
Intuitively, $\l_\a$ ``renormalizes'' the infinity $\Tr_{\bA,\a}(\bm I_{\bA,\a})$. For the type-I factor $\cm_\a=\cb(\ch_{A,\a})\otimes \bm I_{\bA,\a}$, the minimal projection is generally given by $\bm P_\a=\bm P_{A,\a}\otimes \bm I_{\bA,\a}$ where $\bm P_{A,\a}$ projects $\ch_{A,\a}$ to some 1-dimensional subspace. Therefore,
\be
\hat{\Tr}_\a(\bm P_\a)=\l_\a.\label{hatTr}
\ee
$\hat{\Tr}_\a$ defines the renormalized trace on $\cm_\a$ and is unique up rescaling $\l_\a$. 


The definition of the entanglement entropy is a generalization of the discussion in Section \ref{finitedimcase}. Given any density operator $\bm{\rho}\in \cb(\ch)$ normalized by $\Tr(\bm{\rho})=1$ with the standard trace on $\ch$. The diagonal block $\bm{\rho}_{\a\a}$ is a bounded operator on $\ch_{A,\a}\otimes \ch_{\bA,\a}$. We take the partial trace on $\ch_{\bA,\a}$ and define
\be
\bm{\rho}_{A,\a}=p_\a^{-1}\Tr_{\bA,\a}(\bm{\rho}_{\a\a})
\ee
where $p_\a$ is the normalization constant such that $\Tr_{A,\a}(\bm{\rho}_{A,\a})=1$. $\Tr(\bm{\rho})=1$ implies $\sum_\a p_\a =1$. We define the density operator $\hat{\bm \rho}_\cm\in\cm$ by
\be
\hat{\bm \rho}_\cm=\bigoplus_\a p_\a \l_\a^{-1} \hat{\bm \rho}_{\cm,\a},\qquad \hat{\bm \rho}_{\cm,\a}=\bm{\rho}_{A,\a}\otimes \bm{I}_{\bA,\a},
\ee
and it satisfies for any $\bm \co=\oplus_\a (\bm \co_{A,\a}\otimes\bm I_{\bA,\a} )\in\cm$,
\be
\Tr(\bm{\rho}\bm \co)=\sum_\a p_\a\Tr_{A,\a} \lt(\bm{\rho}_{A,\a}\bm\co_{A,\a}\rt)=\sum_\a p_\a\lambda_\a^{-1} \hat{\Tr}_\a\lt (\hat{\bm \rho}_{\cm,\a}\bm\co_\a\rt)=\hat{\Tr}\lt(\hat{\bm \rho}_\cm\bm \co\rt).
\ee
The entanglement entropy of $\bm \rho$ is given by
\be
S(\bm \rho,A)=-\hat{\Tr}\lt(\hat{\bm \rho}_\cm\log \hat{\bm \rho}_\cm\rt)=-\sum_\alpha p_\alpha \log p_\alpha +\sum_\alpha p_\alpha \lt[S(\bm{\rho}_{A,\alpha})+\log\lambda_{\alpha}\rt],\label{entropyformula2}
\ee
where the von Neumann entropy $S(\bm{\rho}_{A,\alpha})$ is defined with the standard trace $\mathrm{Tr}$ on $\ch_{A,\a}\otimes\ch_{\bA,\a}$. Eq.\eqref{entropyformula2} generalizes \eqref{entropyformula} to the infinite-dimensional Hilbert space $\ch$. There exists ambiguity in the definition of entanglement entropy due to the choices of $\l_\a$ in defining the trace. The formula \eqref{entropyformula} corresponds to the choice $\l_\a=1$ for all $\a$.

The sum $\sum_\a$ is an infinite sum in \eqref{entropyformula2}. Since the finite-dimensional Hilbert space $\hat{\ch}$ is a proper subspace of $\ch$, every state that discussed above can be seen as a state in $\ch$. The corresponding density operator only has a finite number of nonzero $p_\a$. Therefore its entanglement entropy computed in $\ch$ using \eqref{entropyformula2} is the same as the one computed in $\hat{\ch}$ using \eqref{entropyformula} for $\l_\a=1$. The ambiguity from $\l_\a$ shifts the bulk entropy $S(\bm{\rho}_{A,\alpha})\to S(\bm{\rho}_{A,\alpha})+\log\l_\a$.


\subsection{Hidden sector Hilbert space}\label{Hidden sector Hilbert space}

In the context of the type I von Neumann algebra from gravitational path integrals \cite{Colafranceschi:2023moh}, every $\l_\a$ is an arbitrary integer interpreted as the dimension of a ``hidden sector Hilbert space'' that is not seen by the von Neumann algebra. The situation is similar in the case of LQG Hilbert space or lattice gauge theory in general: There exists a canonical choice of $\l_\alpha$ being the dimension of the representation of the gauge group at the punctures $\Fs\cap\{\fl_0\}$. Indeed, if we let $\Fs$ to cut every $\fl_0$ into half-links $\fl^{(A)}_0,\fl^{(\bA)}_0$ and do not impose the gauge invariance at the cut, there can be states having two different spins on the pair of half-links: $j_{\fl^{(A)}_0}\neq j_{\fl^{(\bA)}_0}$. The (enlarged) Hilbert space $\ch_{0}$ including these states is a tensor product:
\be
\ch_{0}=\ch_{A}\otimes\ch_{\bA}.
\ee
The Hilbert spaces $\ch_A$ and $\ch_{\bA}$ contains all spin-network states in $A$ and $\bA$ respectively without imposing the gauge invariance on the boundary $\Fs$:
\be
\ch_A=\bigoplus_\a \ch_{A,\a} \otimes V_\a,\qquad \ch_\bA=\bigoplus_\a \ch_{\bA,\a} \otimes V_\a^*,\qquad V_{\a=\{j_{\fl_0}\}} =\bigotimes_{\fl_0}V_{j_{\fl_0}}
\ee 
where $V_j$ is the spin-$j$ irrep of SU(2). The ``hidden sector Hilbert space'' $ V_\a$ carrying the representation of the gauge transformations on the entangling surface  $\Fs$. The states in the hidden sector Hilbert space can be understood as the edge modes at $\Fs$, relating to the discussion in \cite{Lin:2018bud}.

Edge modes represent the degrees of freedom that would typically be unphysical gauge choices but become genuine, observable physical states due to the observer's limited access to the region beyond the horizon \cite{Balachandran:1994up,Carlip:1994gy,LeighEdgeModes}. Intuitively, the boundary $\Fs=\partial A$ is the ``horizon'' for an interior observer who only sees the observables in $\cm$. The presence of the boundary means the observer cannot perform observations or measurements that would verify the gauge constraint holds true across the entire spacetime, particularly in the unobservable region beyond $\Fs$. Since the observer cannot enforce the gauge constraint at $\Fs$, the degrees of freedom associated with the gauge redundancy at this boundary--the edge modes--should not be eliminated. Instead, they should be promoted to physical degrees of freedom that encode information about the region beyond the observer's reach.


The Hilbert space $\ch$ of gauge invariant states only contains states with $j_{\fl^{(A)}_0} = j_{\fl^{(\bA)}_0}$, so $\ch$ can be embedded as a proper subspace of $\ch_0$. A normalized state $|\psi_\a\rangle\in\ch_{A,\a}\otimes\ch_{\bA,\a}$ is equivalent to\footnote{$\Psi_\alpha$ adds bi-valent intertwiners to the virtual nodes at the punctures. $\Psi_\a$ is equivalent to $\psi_\a$ due to the cylindrical consistency.}
\be
|\Psi_\a\rangle=|\psi_\a\rangle\bigotimes_{\fl_0}\lt(\sum_{m_{\fl_0}}\frac{\mid j_{\fl_0},m_{\fl_0}\rangle_A \mid j_{\fl_0},m_{\fl_0}\rangle_{\bA}}{\sqrt {2j_{\fl_0}+1}}\rt)\in\ch_0\ .
\ee 
where $| j_{\fl_0},m_{\fl_0}\rangle_A\in V_{j_{\fl_0}}$ and $| j_{\fl_0},m_{\fl_0}\rangle_{\bA}\in V_{j_{\fl_0}}^*$, and $\Psi_\a$ is gauge invariant and normalized in $\ch_0$. 

Any gauge invariant operator $\bm\co_\a=\bm\co_{A,\a}\otimes\bm I_{\bA,\a}$ with $\a=\{j_{\fl_0}\}$ corresponds to an operator $\bm\co_{A,\a}\otimes \bm I_{V_\a}$ on $\ch_A$, where $\bm I_{V_\a}$ is the identity operator on $V_\a$. This correspondence maps the von Neumann algebra $\cm$ into a subalgebra in $\cb(\ch_A)$. In particular, the minimal projection $\bm P_\a$ onto 1-dimension subspace $V_0$ in $\ch_{A,\a}$ is mapped to the projection $\bm P_\a\otimes \bm I_{V_\a}$ onto $V_0\otimes V_\a$ on $\ch_A$. Denote by $\mathrm{tr}_A$ the standard trace on $\ch_A$, we obtain
\be
\mathrm{tr}_A\lt(\bm P_\a\otimes \bm I_{\{j_{\fl_0}\}}\rt)=\dim(V_\a)=\prod_{\fl_0}(2j_{\fl_0}+1).
\ee
Therefore, the renormalized trace $\hat{\Tr}_\a(\bm\co_\a)$ defined in \eqref{hatTr84} and \eqref{hatTr} equals to the standard trace $\mathrm{tr}_A\lt(\bm\co_{A,\a}\otimes \bm I_{V_\a}\rt)$ on $\ch_A$ when 
\be
\l_\a=\prod_{\fl_0}(2j_{\fl_0}+1).\label{canonChoice}
\ee
This fixes the ambiguity $\l_\a$ in the definition of the renormalized trace $\hat{\Tr}$.

We denote the standard trace on $\ch_0$ by $\mathrm{tr}$. The gauge invariant state $\bm\rho=\sum_{\a,\a'}p^{1/2}_\a p^{1/2}_{\a'}|\psi_\a\rangle\langle\psi_{\a'}|$ on $\ch$ is equivalent to $\rho=\sum_{\a,\a'}p^{1/2}_\a p^{1/2}_{\a'}|\Psi_\a\rangle\langle\Psi_{\a'}|$ on $\ch_0$. The reduced density matrix is given by
\be
\rho_A=\mathrm{tr}_{\bA}\rho=\bigoplus_\a p_\a\sum_{i}\langle \a, i_{\bA}\mid\Psi_\a\rangle \langle\Psi_\a\mid \a, i_{\bA}\rangle=\bigoplus_\a \lt[ p_\a \l_\a^{-1}\bm \rho_{A,\a}\otimes \bm I_{V_\a}\rt].
\ee
where $i_{\bA}$ labels the orthonormal basis in $\ch_{\bA,\a}\otimes V_\a^*$, so we have the correspondence between $\hat{\bm\rho}_{\cm}=\oplus_\a \lt[p_\a \l_\a^{-1}\bm\rho_{A,\a}\otimes \bm I_{\bA,\a}\rt]$ and $\rho_A$. $S(\bm\rho,A)$ in \eqref{entropyformula2} equals to the standard von Neumann entropy on $\ch_0$ due to the correspondence between $\hat{\Tr}$ and $\mathrm{tr}$:
\begin{eqnarray}
S(\bm\rho,A)=-\hat{\Tr}\lt(\hat{\bm \rho}_\cm\log \hat{\bm \rho}_\cm\rt)=-\mathrm{tr}_A(\rho_A\log \rho_A)=-\sum_\alpha p_\alpha \log p_\alpha +\sum_\alpha p_\alpha \lt[S(\bm{\rho}_{A,\alpha})+\sum_{\fl_0}\log\lt(2j_{\fl_0}+1\rt)\rt].\label{entropywithdim}
\end{eqnarray}
 
The term $\sum_{\fl_0}\log\lt(2j_{\fl_0}+1\rt)$ in $S(\bm\rho, A)$ coming from the canonical choice of $\l_\a$ is precisely the entanglement entropy obtained in \cite{Donnelly:2008vx}. Indeed, given any gauge invariant spin-network state with all the spins fixed, we have $p_\a=1$ only for one $\a$, and $S(\bm{\rho}_{A,\alpha})=0$ since the state is a pure tensor product of intertwiners, so only the term $\sum_{\fl_0}\log\lt(2j_{\fl_0}+1\rt)$ survives in $S(\bm\rho, A)$.

The embedding $\ch\subset \ch_0$ and the canonical choice \eqref{canonChoice} have the physical interpretation in terms of quantum error correction code. The detailed discussion of this aspect is given in \cite{HanTobin}.

We repeat the analysis in Section \ref{Fixed-area state and geometrical area law} with the entropy formula \eqref{entropywithdim} for the state \eqref{Psisum1} with $S(\bm{\rho}_{A,\alpha})=0$. The modification is only to change the degeneracy from $G(j)=1$ in \eqref{Dstatecount} to $G(j)=2j+1$ (see also \eqref{factorbulk} and the discussion there, since $\sum_{\fl_0}\log\lt(2j_{\fl_0}+1\rt)$ plays the same role as the bulk entropy discussed there). The entanglement entropy can be computed by using $P(k)=\prod_{k\neq 0}(k+1)^{n_k}$ in the state-counting \eqref{cn(a)}. The result is given by the area law with a different coefficient:
\be
S(\bm{\rho}_c,A) =  \frac{\b}{8\pi\g \ell_P^2}\Ar(\Fs)+O(1),\label{eq104}
\ee
where $\b/(2\pi)=0.274067\cdots$ is given by 
\be
\sum_{k=1}^{\infty}\left(k+1\right)e^{-(\b/2)\sqrt{\left(k+1\right)^{2}-1}}=1,
\ee
the same as $\b=2s_0$ in \eqref{blackholeBetavalue} and the one for the SU(2) black hole entropy.

The state-counting with $P(k)=\prod_{k\neq 0}(k+1)^{n_k}$ here does not impose any constraint on the spin profile other than the area constraint. In contrast, $P(k)$ in \eqref{eq75} imposes both the triangle inequality and the constraint that the sum of all $k$'s is an even number are imposed. Moreover, as the dimension of the intertwiner space, $P(k)$ in \eqref{eq75} is generally less than $\prod_{k\neq 0}(k+1)^{n_k}$. However, these two situations gives exactly the same leading order area-law from the state-counting. The constraints and having less states in the intertwiner space only affect the logarithmic correction and $O(1)$ contribution. In the case that $A$ contains many nodes and $\bm\rho_c$ has the factorized bulk state as in \eqref{Psisum1}, the entanglement entropy corresponds to the state-counting \eqref{cn(a)} with $P(k)$ in between the above two cases, i.e., the sum of all $k$'s still has to be an even number, while the triangle inequality is alleviated, and $P(k)=\prod_{k\neq 0}(k+1)^{n_k}$ for the allowed $k$'s. Since the number of states is in between the above two cases, the resulting entanglement has the same leader order area-law as in \eqref{eq104}.

\section*{Acknowledgments}

The author acknowledges some useful discussions with Hongguang Liu, Qiaoyin Pan, and Sean Tobin. The author receives supports from the National Science Foundation through grants PHY-2207763 and PHY-2512890 and from a sponsorship of renewed research stay in Erlangen from the Humboldt Foundation.

\appendix

\section{Entanglement entropy of random state}\label{random average}

Given the finite-dimensional Hilbert space $\ch=\ch_A\otimes \ch_{\bA}$ and the random state 
\be
\bm{\rho}=|\psi_{A,\bA}\rangle \langle \psi_{A,\bA}|=\cu |\psi^{(0)}_{A,\bA}\rangle\langle \psi^{(0)}_{A,\bA}|\cu^\dagger , 
\ee
the second Renyi entropy $S_2(\bm \rho_A)=-\log [\Tr_A(\bm \rho_A^2)/(\Tr_A\bm \rho_A)^2]$ reformulated by using the ``swap trick''
\be
S_2(\bm \rho_A)=-\log\frac{Z_1}{Z_0}, \qquad Z_1=\Tr\lt[(\bm \rho\otimes\bm \rho)\cf_A\rt],\qquad Z_0=\Tr\lt[\bm \rho\otimes\bm \rho\rt],
\ee
where $\Tr$ is the trace over $\ch_A\otimes \ch_{\bA}$, and $\cf_A$ is the swap operator on $(\ch_A\otimes \ch_{\bA})^{\otimes 2}$: $\cf_A$ swaps the states of the two copies in the region A, i.e. $\cf_A |\psi_A\rangle\otimes| \phi_{\bA}\rangle \otimes |\psi'_A\rangle\otimes |\phi'_{\bA}\rangle=|\psi_A'\rangle\otimes| \phi_{\bA}\rangle \otimes |\psi_A\rangle\otimes |\phi'_{\bA}\rangle$. The random average $\langle S_2(\bm \rho_A) \rangle := \int\rmd \cu\ S_2(\bm \rho_A)$ can be approximated by averaging separately $Z_1$ and $Z_0$
\be
\langle S_2(\bm \rho_A) \rangle\simeq -\log\frac{\langle Z_1\rangle}{\langle Z_0\rangle }.\label{averageZZ}
\ee
The validity of this formula will be proven in a moment.

We first apply \eqref{averageZZ} to computing the averaged entropy and use the following formula to compute $\langle Z_1\rangle$ and $\langle Z_0\rangle$, 
\be
\langle\bm{\rho}\otimes\bm{\rho}\rangle=\int \rmd\cu \, \bm{\rho}\otimes\bm{\rho}=\frac{\ci+\cf}{d_\ch+d^2_\ch},\label{averagerhorho}
\ee
where $\rmd\cu$ is the Haar measure, and $\ci$ and $\cf$ are identity operator and swap operator on $\ch\otimes\ch$, i.e.
\be
\ci|\psi\rangle \otimes|\phi\rangle=|\psi\rangle \otimes|\phi\rangle,\qquad \cf|\psi\rangle \otimes|\phi\rangle=|\phi\rangle \otimes|\psi\rangle,\qquad \psi,\phi \in\ch
\ee
We denote by $d_A=\dim (\ch_A)$, $d_{\bA}=\dim(\ch_{\bA})$, and $d_{\ch}=d_A d_\bA=\dim(\ch)$. FIG.\ref{graphicComp} introduces some graphic notations for the random state $|\psi_{A,\bA}\rangle\in\ch_A\otimes\ch_\bA$ and the density matrix$\bm{\rho}$, as well as the formula \eqref{averagerhorho} in terms of the graphic notations. FIG.\ref{Z1Z0} applies the graphic notations to the computations of the averages $\langle Z_1\rangle$ and $\langle Z_0\rangle$. As a result, we obtain that for $d_A,d_\bA\gg 1$,
\be
\langle S_2(\bm \rho_A) \rangle\simeq -\log\lt(\frac{1}{d_A}+\frac{1}{d_\bA}\rt)\simeq \log \lt[\mathrm{min}\{d_A,d_\bA\}\rt].
\ee
The von Neumann entropy is lower bounded by the second Renyi entropy: $S(\bm \rho_A)\geq S_2(\bm \rho_A)$, and $\log \lt[\mathrm{min}(d_A,d_\bA)\rt]$ is already the maximal value of the entanglement entropy on $\ch_A\otimes \ch_\bA$, so we obtain
\be
\langle S(\bm \rho_A) \rangle \simeq \log \lt[\mathrm{min}(d_A,d_\bA)\rt].
\ee

\begin{figure}[h]
\centering
    \includegraphics[scale=0.4]{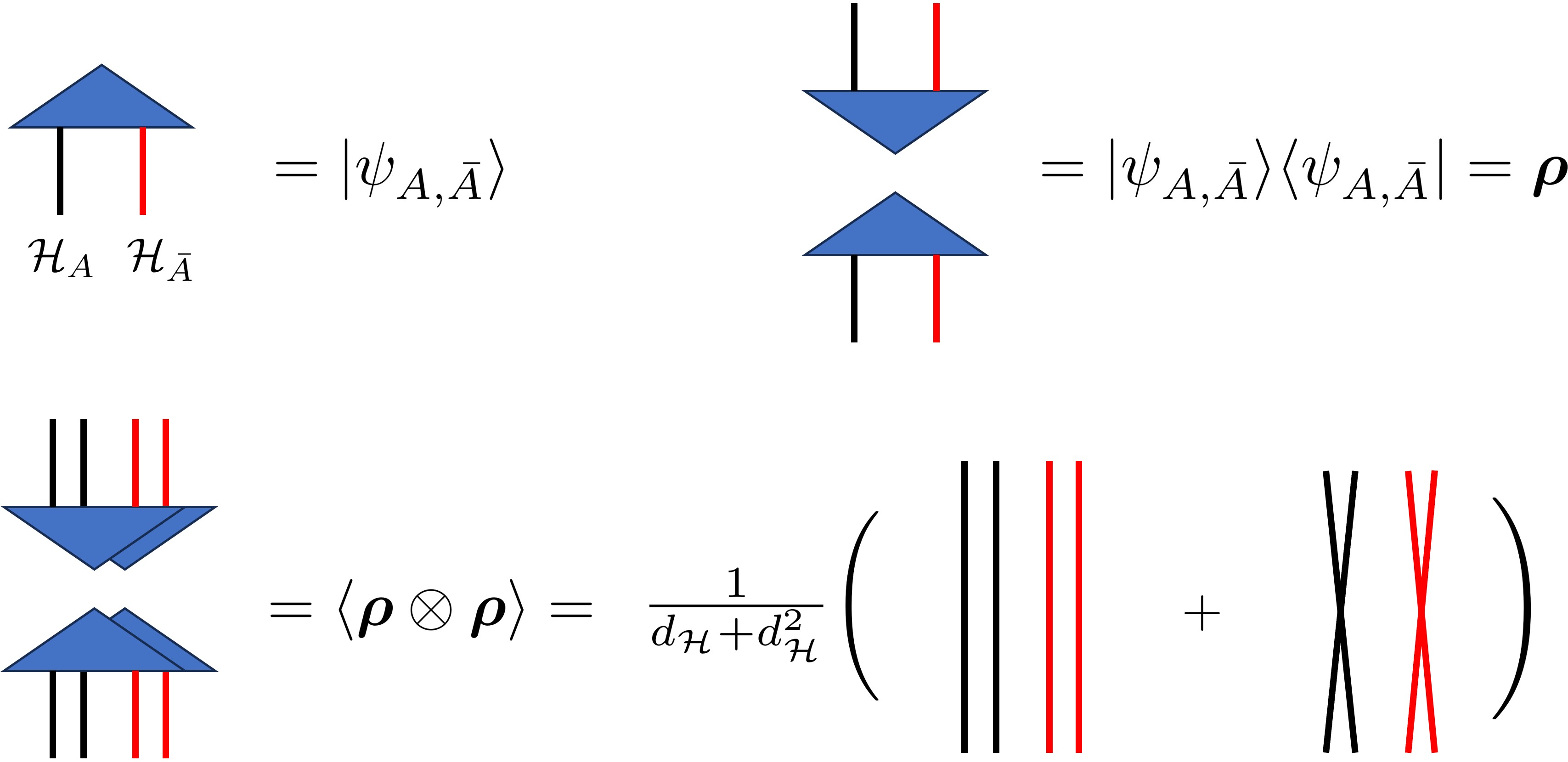}
    \caption{Graphic notations for $|\psi_{A,\bA}\rangle$, $\bm{\rho}$, and \eqref{averagerhorho}.}
    \label{graphicComp}
\end{figure}

\begin{figure}[h]
\centering
    \includegraphics[width=1\textwidth]{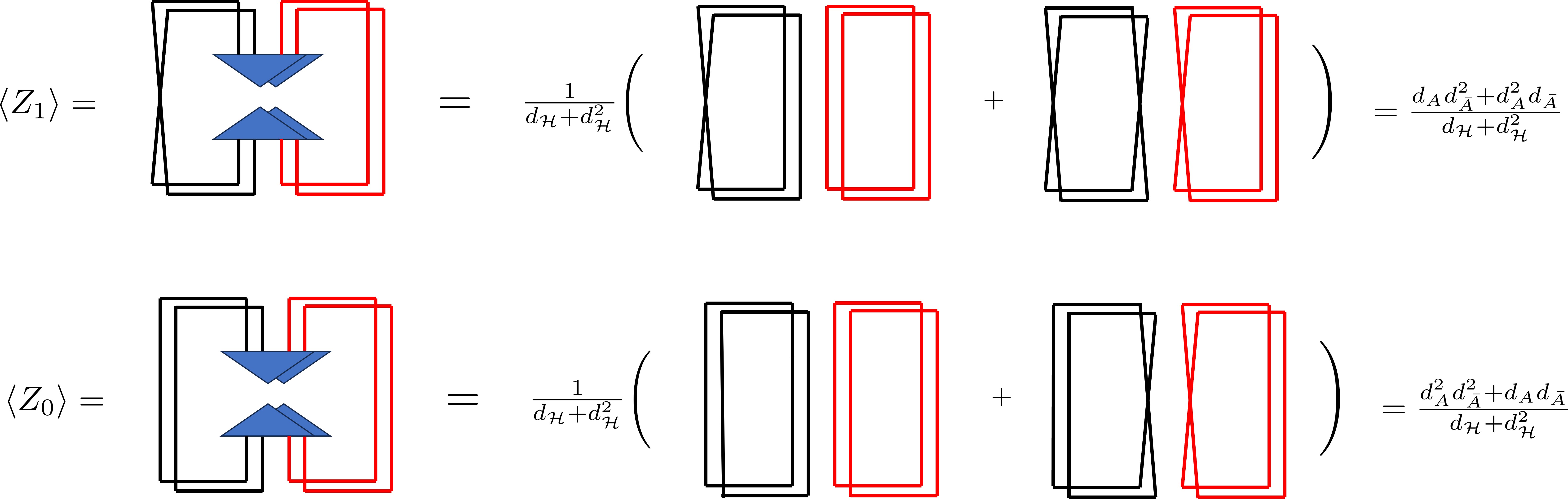}
    \caption{The graphic computations of the averages $\langle Z_1\rangle$ and $\langle Z_0\rangle$.}
    \label{Z1Z0}
\end{figure}

Now we show that \eqref{averageZZ} is indeed valid, namely, $S_2(\bm{\rho})$ concentrates at $\langle S_2(\bm{\rho})\rangle$ computed from \eqref{averageZZ} with a high probability for large dimension. The idea of derivation is similar to  \cite{Qi1}. We consider the fluctuation
\be
\frac{\langle(Z_1-\langle Z_1\rangle)^2\rangle}{\langle Z_1\rangle^2}=\frac{\langle Z_1^2 \rangle}{\langle Z_1\rangle^2}-1.
\ee
We compute the average $\langle Z_1^2 \rangle=\langle \Tr[(\bm \rho\otimes\bm \rho)\cf_A]^2 \rangle$ by using
\be
\langle\bm{\rho}^{\otimes 4}\rangle=\int \rmd\cu \, \bm{\rho}^{\otimes 4}=\frac{1}{\cc_4}\sum_{\sig\in{\rm Sym}_4} \sig ,\label{averagerho4}
\ee
where $\cc_m=(d_\ch+m-1)!/(d_\ch-1)!$, and the sum is over all permutations $\sig$ acting on $\ch^{\otimes 4}$. If we denote by $\L=(d_A^{-1}+d_\bA^{-1})^{-1}$, we obtain that
\be
\frac{\langle Z_1^2 \rangle}{\langle Z_1\rangle^2}-1=O(\L^{-1}).
\ee
By Markov's inequality
\be
{\rm Prob}\lt(\lt|\frac{Z_1}{\langle Z_1\rangle}-1\rt|\geq \frac{\delta}{4}\rt)\leq \frac{\lag\lt(\frac{Z_1}{\langle Z_1\rangle}-1\rt)^2\rag}{(\delta/4)^2}=O\lt(\delta^{-2}\L^{-1}\rt)
\ee
A similar bound holds for ${\rm Prob}\lt(\lt|\frac{Z_0}{\langle Z_0\rangle}-1\rt|\geq \frac{\delta}{4}\rt)$. The mean that with the probability of at least $1-O\lt(\delta^{-2}\L^{-1}\rt)$, we have $\lt|\frac{Z_a}{\langle Z_a\rangle}-1\rt|\leq \frac{\delta}{4}$, $a=0,1$. Then
\be
\lt|S_2(\bm \rho_A) - \langle S_2(\bm \rho_A)\rangle \rt|&=&\lt|\log \frac{Z_1}{Z_0}-\log \frac{\langle Z_1\rangle}{\langle Z_0\rangle }\rt|=\lt|\log \frac{Z_1}{\langle Z_1\rangle}-\log \frac{Z_0}{\langle Z_0\rangle }\rt|\nonumber\\
&\leq &\lt|\log \frac{Z_1}{\langle Z_1\rangle}\rt|+\lt|\log \frac{Z_0}{\langle Z_0\rangle }\rt|\leq \delta
\ee
For any small $\delta >0$, there is a large probability $1-O\lt(\delta^{-2}\L^{-1}\rt)$ for suitably large $d_A$ or $d_\bA$, such that $\lt|S_2(\bm \rho_A) - \langle S_2(\bm \rho_A)\rangle \rt|\leq \delta$.

\bibliographystyle{jhep}
\bibliography{muxin}

\end{document}